\documentclass[a4paper,twocolumn,11pt]{quantumarticle}
\pdfoutput=1
\usepackage[utf8]{inputenc}
\usepackage[english]{babel}
\usepackage[T1]{fontenc}
\usepackage{amsmath}
\usepackage{hyperref}
\usepackage[ruled,vlined]{algorithm2e}
\usepackage{enumitem}
\usepackage{amsthm}
\usepackage{amssymb}
\usepackage{graphicx}
\usepackage{subfigure}

\usepackage{tikz}
\usepackage{lipsum}
\usepackage[numbers,comma,sort&compress]{natbib}

\newtheorem{theorem}{Theorem}

\newtheorem{lemma}{Lemma}

\newcommand{\thm}[1]{\hyperref[thm:#1]{Theorem~\ref*{thm:#1}}}
\newcommand{\cor}[1]{\hyperref[cor:#1]{Corollary~\ref*{cor:#1}}}
\newcommand{\defn}[1]{\hyperref[defn:#1]{Definition~\ref*{defn:#1}}}
\newcommand{\lem}[1]{\hyperref[lem:#1]{Lemma~\ref*{lem:#1}}}
\newcommand{\prop}[1]{\hyperref[prop:#1]{Proposition~\ref*{prop:#1}}}
\newcommand{\fig}[1]{\hyperref[fig:#1]{Figure~\ref*{fig:#1}}}
\newcommand{\tab}[1]{\hyperref[tab:#1]{Table~\ref*{tab:#1}}}
\newcommand{\algo}[1]{\hyperref[algo:#1]{Algorithm~\ref*{algo:#1}}}
\renewcommand{\sec}[1]{\hyperref[sec:#1]{Section~\ref*{sec:#1}}}
\newcommand{\append}[1]{\hyperref[append:#1]{Appendix~\ref*{append:#1}}}
\newcommand{\fac}[1]{\hyperref[fac:#1]{Fact~\ref*{fac:#1}}}
\newcommand{\lin}[1]{\hyperref[lin:#1]{Line~\ref*{lin:#1}}}

\def\>{\rangle}
\def\<{\langle}

\def\map#1{\mathcal #1}
\def\set#1{\mathsf{#1}}

\newcommand{\R}{\mathbb{R}}

\def\Tr{\operatorname{Tr}}\def\:{\hbox{\bf:}}

\def\?#1{\if.#1{}\else#1\fi}

\begin{document}

\title{Quantum autoencoders for communication-efficient  quantum cloud computing}

\author{Yan Zhu}
\affiliation{QICI Quantum Information and Computation Initiative, Department of Computer Science,
The University of Hong Kong, Pokfulam Road, Hong Kong}
\affiliation{HKU-Oxford Joint Laboratory for Quantum Information and Computation}
\author{Ge Bai}
\affiliation{QICI Quantum Information and Computation Initiative, Department of Computer Science,
The University of Hong Kong, Pokfulam Road, Hong Kong}
\affiliation{HKU-Oxford Joint Laboratory for Quantum Information and Computation}
\author{Yuexuan Wang}
\affiliation{College of Computer Science and Technology, Zhejiang University, Hangzhou, China}
\author{Tongyang Li}\thanks{Corresponding author, email: tongyangli@pku.edu.cn}
\affiliation{Center on Frontiers of Computing Studies, Peking University}
\affiliation{School of Computer Science, Peking University}
\affiliation{Center for Theoretical Physics, Massachusetts Institute of Technology}
\author{Giulio Chiribella}\thanks{Corresponding author, email: giulio@cs.hku.hk}
\affiliation{QICI Quantum Information and Computation Initiative, Department of Computer Science,
The University of Hong Kong, Pokfulam Road, Hong Kong}
\affiliation{Department of Computer Science, University of Oxford, Parks Road, Oxford OX1 3QD, United Kingdom}
\affiliation{Perimeter Institute For Theoretical Physics, 31 Caroline Street North, Waterloo N2L 2Y5, Ontario, Canada}
\affiliation{The University of Hong Kong Shenzhen Institute of Research and Innovation, Yuexing 2nd Rd Nanshan, Shenzhen 518057, China}

\maketitle

\begin{abstract}
In the model of quantum cloud computing, the server executes a computation on the quantum data provided by the client.
In this scenario, it is important to reduce the amount of quantum communication between the client and the server. A possible approach is to transform the desired computation into a compressed version that acts on a smaller number of qubits, thereby reducing the amount of data exchanged between the client and the server.  Here we propose quantum autoencoders for quantum gates (QAEGate) as a method for compressing quantum computations. We illustrate it in concrete scenarios of single-round and multi-round communication and validate it through numerical experiments.
A bonus of our method is it does not reveal any information about the server's computation other than the information present in the output.

\end{abstract}

\section{Introduction}

Over the past decade, quantum computing technology underwent a rapid series of advances.  Both the size and power of quantum computers have been steadily increasing over the years,  recently entering a new regime of  "quantum supremacy"~\cite{Preskill2018NISQ}, in which the output of the quantum computations can be barely reproduced by the world's best supercomputers.   Milestone achievements are the "quantum supremacy" demonstrations  by Google~\cite{arute2019supremacy} and USTC~\cite{zhong2020quantum}, using superconducting qubits and photonic qubits, respectively.

A promising direction in quantum computing is the study of cloud computing scenarios, where a client requests a remote server to perform some desired quantum operations \cite{broadbent2009universal, barz2012demonstration, morimae2013blind}. However, limits on the amount of quantum communication between client and server constraint the size of the computations that can be effectively implemented in quantum cloud computing. 

In this paper, we address the problem of reducing the amount of communication needed by the server to perform a quantum operation on quantum data provided by the client. The operation belongs to a parametric family of quantum gates known to both parties, but the specifications of the gate are known only to the server.    Our objective here is to maximize the accuracy of the executed quantum operation when the capacity of the communication link and the total number of qubits exchanged between the server and the client are limited. 

Our main contribution is a method for compressing a parametric family of quantum gates, turning it into another gate family acting on a smaller number of qubits. Our method is based on autoencoders, a type of neural networks that have been very successful in classical machine learning \cite{schmidhuber2015deep}, and have been recently used for the compression of quantum states   \cite{romero2017quantum}. 
We introduce a  quantum gate autoencoder (QAEGate), providing the first quantum machine learning model that takes gates, rather than states, as inputs.

Our method also addresses the problem of minimizing the amount of information revealed to the client. In many situations,  the server would not like to disclose any more  information about the operation other than the unavoidable information revealed  by the application of the operation on the client's input. Our method achieves this feature by constructing a blind compression protocol, independent of the specifications of the operation.  

Technically, the training of our QAEGate model is based on stochastic gradient descent \cite{bottou2012stochastic}. We prove that the training is guaranteed to converge in polynomial time in the size of the initial operation. We then conduct numerical experiments that show the effectiveness of QAEGate in various settings including single-round and multi-round communication between the client and the server.

The remainder of the paper is organized as follows. In \sec{relat}, we introduce some related works about quantum cloud computing and quantum autoencoders. Some preliminaries on autoencoders and quantum supermaps are provided in \sec{prelim}. \sec{task} is devoted to introduce the quantum cloud computing task discussed in this paper. The structure and training details of our proposed QAEGate model are covered in \sec{QAEGate}. Then, we explain how to apply our method to concrete scenarios of quantum cloud computing in \sec{app} and conduct some numerical experiments in \sec{exper}. We conclude this paper in \sec{conclu}, with discussions on future directions for our method.

\section{Related work}\label{sec:relat}
Classical cloud computing  involves the execution of a computation on a remote server \cite{armbrust2010view}. In the quantum version of this scenario \cite{barz2012demonstration}, a client  provides an input state  and asks  a remote server  to perform a sequence of quantum gates on it.  Two important issues are (1) to cope with the limited amount of quantum communication available in realistic scenarios, and (2)  to minimize the leakage of the information about the server's computation.  These two issues have been previously approached with techniques from quantum Shannon theory  \cite{yang2020communication} and quantum cryptography  \cite{sheng2017distributed}. Related  issues have also been considered in blind delegated quantum computation \cite{broadbent2009universal}, where the goal is to guarantee  data confidentiality on  the client side rather than the server side.

Quantum autoencoders for quantum states were proposed by \cite{romero2017quantum} and has been applied to quantum state compression and denoising quantum data \cite{bondarenko2020quantum,achache2020denoising}. The quantum model used in these works has a  similar structure to classical autoencoders, taking as  input  a quantum state represented by a fixed-length vector. This model, however, cannot be extended from quantum states to quantum gates, which in general are provided as black boxes. To convert a gate into a quantum state, one would have to apply it on a fixed input state. However, such conversion is not reversible due to the quantum No-programming theorem \cite{nielsen1997programmable, yang2020optimal}.

\section{Preliminaries}\label{sec:prelim}

\paragraph{Autoencoders.} An autoencoder is composed of an encoder and a decoder. One of its main application is dimensionality reduction, where the encoder maps a high-dimensional vector input $x$  to a low-dimensional representation $h$, and the decoder maps $h$ back to a reconstructed high-dimensional vector $x'$. In the training process, the autoencoder is optimized so that the high-dimensional vector $x'$  is as close as possible to the original high-dimensional vector $x$. If $x'$ and $x$ are close enough, the overall protocol provides a faithful
% \bg{good->faithful} 
compression of the original input $x$ into the low dimensional vector $h$. The structure of a typical autoencoder is depicted in \fig{autoencoder}.

\begin{figure}[ht]
% \vskip 0.2in
% \begin{center}
\centering
\centerline{\includegraphics[width=80mm]{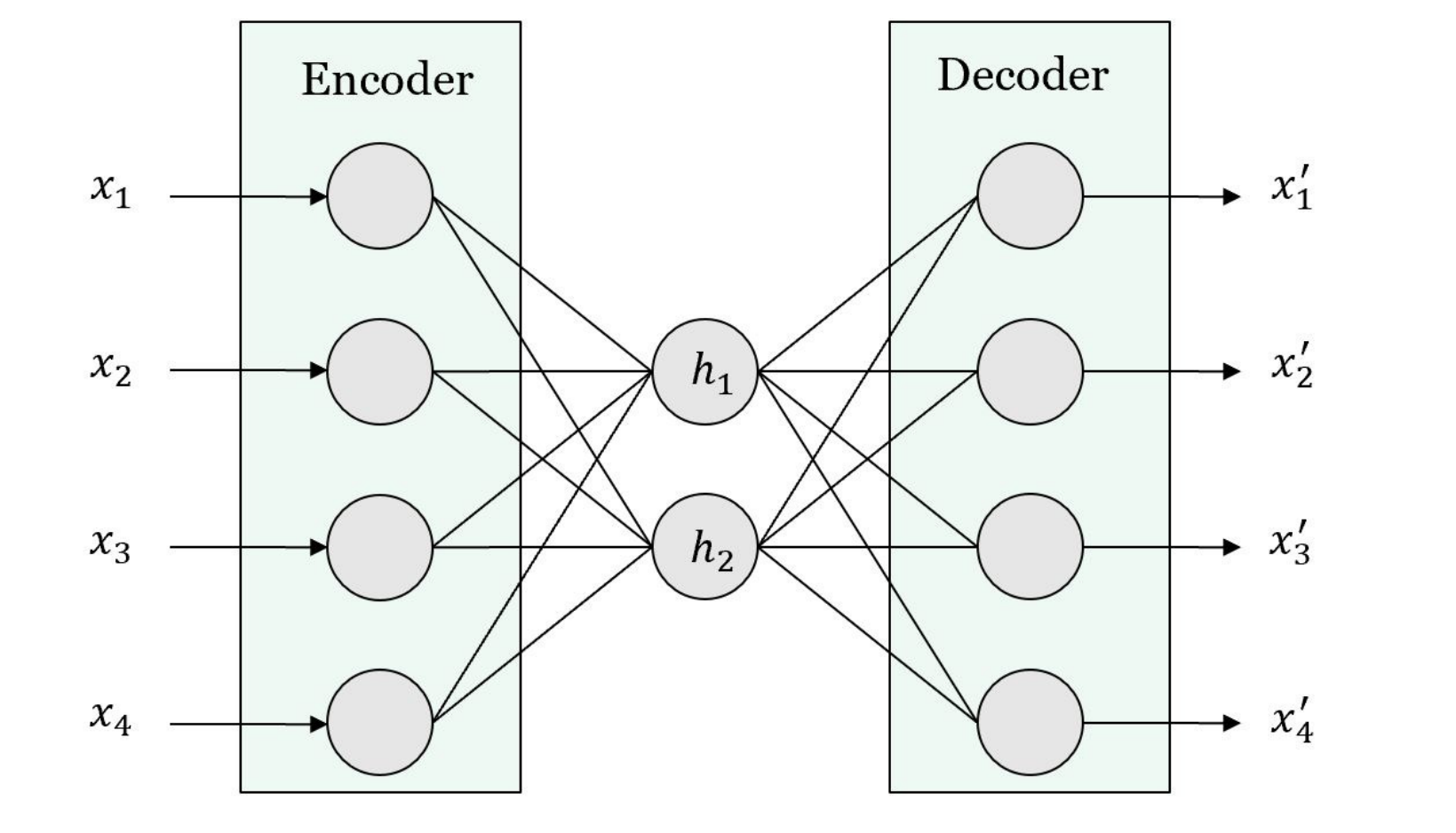}}
\caption{Structure of a typical autoencoder. In this specific example, the  autoencoder compresses a 4-dimensional vector into a 2-dimensional vector.}
\label{fig:autoencoder}
% \end{center}
% \vskip -0.2in
\end{figure}

\paragraph{Quantum supermaps.}

A general quantum operation is represented by a quantum channel $\mathcal{C}$, which is a map from density operators to density operators \cite{nielsen2002quantum}.
A quantum supermap \cite{chiribella2008transforming,chiribella2009theoretical} $\tilde{\mathcal{S}}$ maps a quantum channel $\mathcal{C}$ into a quantum channel $\mathcal{C}'$ as $\mathcal{C}' = \tilde{\mathcal{S}}(\mathcal{C})$. Every quantum supermap can be represented by a quantum circuit in \fig{supermap} \cite{chiribella2008transforming}, where the input
quantum operation $\mathcal{C}$ sends states in $\mathcal H_{in}$ to states in $\mathcal H_{out}$ and the output quantum operation $\mathcal{C}'$ sends states on $\mathcal K_{in}$ to states in $\mathcal K_{out}$. The supermap is realized by two maps $\mathcal{V}$ and $\mathcal{W}$ located at the input and at the output ports of the quantum operation $\mathcal{C}$ respectively. 

\begin{figure}[ht]
% \vskip 0.2in
% \begin{center}
\centering
\centerline{\includegraphics[width=70mm]{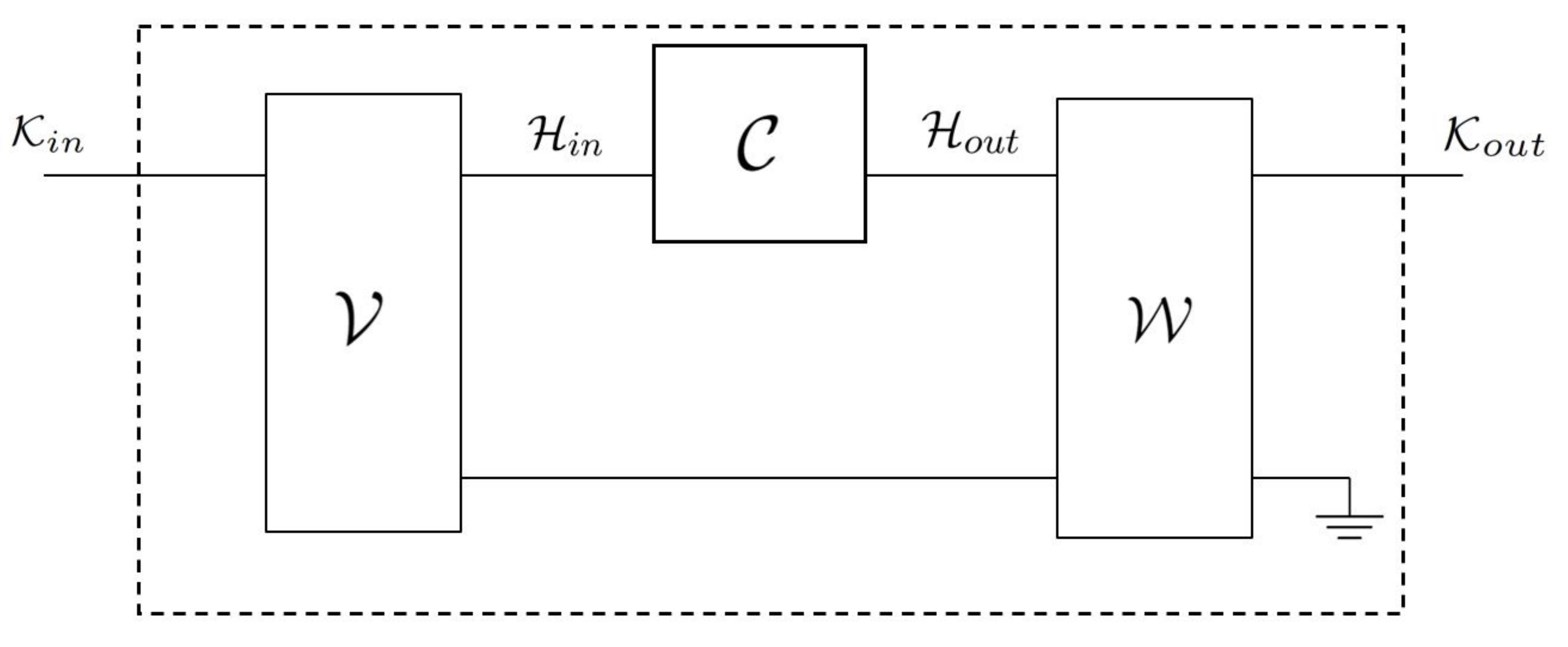}}
\caption{Quantum supermap.}
\label{fig:supermap}
% \end{center}
% \vskip -0.2in
\end{figure}

\section{The Quantum  Cloud Computing Task} \label{sec:task}
In this paper we will focus on a basic scenario of quantum cloud computing,  depicted in \fig{task1}. A parametric family of $n$-qubit quantum gates $\{U_x\}_{x\in \set X}$ with parameter $x$ represents a set of possible quantum computations. A server is able to implement a gate $U_x$,  unknown to the client. The client chooses an $n$-qubit input state $\rho$ and asks  the server  to  apply the gate $U_x$ on it, thus obtaining the state $\map{U}_{x}(\rho):=U_x \rho U_x^\dag$.  At the same time, the server wants to keep $x$ confidential, avoiding unnecessary information leaked to the client during the communication.
 
A trivial way to achieve the above task is to have the sender send the quantum state $\rho$ to the server, who performs the requirexx gate, anxx  sends back the output state $\map{U}_{x}(\rho)$. However, the quantum communication is an expensive resource, and it is often limited in realistic applications, making it difficult to  transmit a full $n$-qubit quantum state back and forth between the client and the server. Here we consider the scenario where the capacity of the quantum communication link is bounded by  $a\le n$ qubits.

In order to cope with  the bottleneck on the amount of quantum communication, we design a pair of quantum circuits, including an encoder $\mathcal{E}$ implemented by the server and a decoder $\mathcal D$ implemented by the client. The encoder $\mathcal E$ serves as a quantum supermap transforming an $n$-qubit quantum channel to an $a$-qubit quantum channel, while the decoder $\mathcal D$ does the opposite, recovering the $n$-qubit quantum channel from the output of the encoder. Since the decoder is implemented by the client, who has no knowledge of the parameter $x$,   the supermap  $\mathcal D$ must be independent of $x$.    On the other hand, the encoder is implemented by the server, and could in principle depend on $x$. However, a dependence on $x$ may result into a leakage of  information to the client.  For this reason, we also require the encoder $\mathcal E$ to be independent of $x$. We will later show that  this choice reaches the best possible confidentiality in the ideal situation where  $U_x$ is accurately implemented.

\begin{figure}[ht]
% \vskip 0.2in
% \begin{center}
\centering
\centerline{\includegraphics[width=80mm]{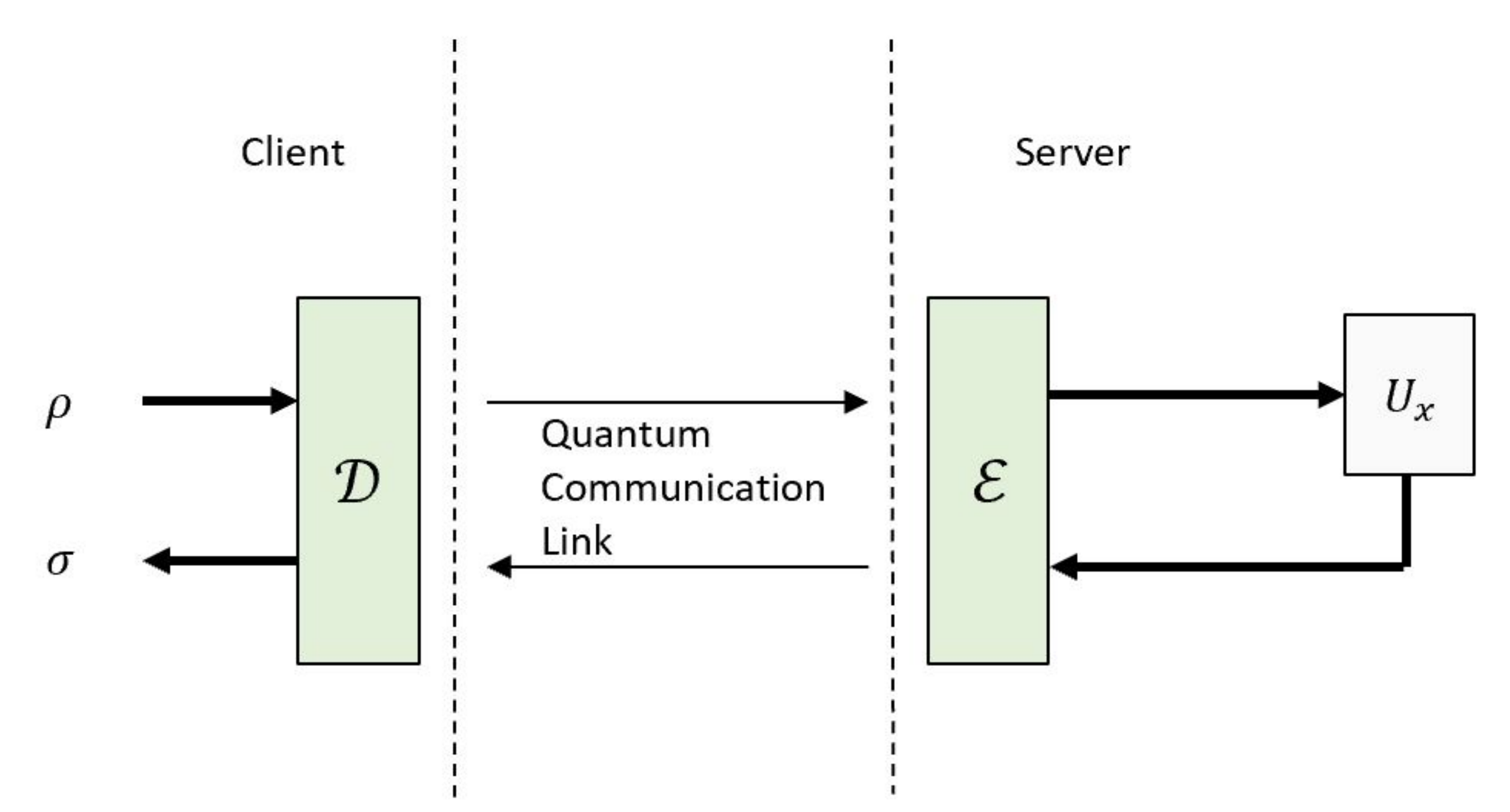}}
\caption{The basic scenario of quantum cloud computing. The client expects to apply a quantum gate $U_x$ on its own state $\rho$ by communicating with a server over a communication link of limited capacity.} 
% {\color{blue} describe what the scenario is about}
\label{fig:task1}
% \end{center}
% \vskip -0.2in
\end{figure}

\section{Quantum Autoencoders for Quantum Gates}\label{sec:QAEGate}

In this section we  consider a basic scenario of quantum cloud computing with one round  of communication between client and server.  In this scenario, we develop a model of quantum autoencoders for quantum gates (QAEGate). In the following, we will introduce the structure of the proposed model and demonstrate its implementation and training through numerical experiments.  Furthermore, we prove that stochastic gradient descent has convergence guarantee in the training of our QAEGate model. We will extend our model to more scenarios in \sec{app}.

\subsection{Structure}
Our model consists of an encoder and a decoder. The encoder produces an encoded quantum gate by inserting the gate $U_x$ into a suitable quantum circuit, initializing $n-a$ input qubits to a fixed state $|0\>$, and discarding $n-a$ output qubits. The result is a (generally noisy) quantum channel acting on $a$ qubits. In turn, the decoder converts the $a$-qubit channel back into an $n$-qubit channel, which aims to approximate the initial gate. The structure of the encoder and the decoder are shown in \fig{implementation}.

To construct the encoding and decoding circuits we use parameterized unitary operators, which have been successfully employed to build variational quantum circuits \cite{cong2019quantum, farhi2014quantum}.  
Our construction  is depicted in \fig{implementation}.
   The encoder consists of  two parameterized unitary operators, denoted by $U^{\boldsymbol{\theta_{\rm le}}}$ and $U^{\boldsymbol{\theta_{\rm re}}}$,  depending on  parameters $\boldsymbol{\theta_{\rm le}} = (\theta_{le_1},\theta_{le_2},...,\theta_{le_m})$  and $\boldsymbol{\theta_{\rm re}} = (\theta_{re_1},\theta_{re_2},...,\theta_{re_m})$. The gates $U^{\boldsymbol{\theta_{\rm le}}}$ and $U^{\boldsymbol{\theta_{\rm re}}}$  are placed on the left  and right of   the original quantum gate, respectively. The decoder consists of  two unitary operators $U^{\boldsymbol{\theta_{\rm ld}}}$ and $U^{\boldsymbol{\theta_{\rm rd}}}$, depending on parameters  $\boldsymbol{\theta_{\rm ld}} = (\theta_{ld_1},\theta_{ld_2},...,\theta_{ld_m})$ and $\boldsymbol{\theta_{\rm rd}} = (\theta_{rd_1},\theta_{rd_2},...,\theta_{rd_m})$.  The gates $U^{\boldsymbol{\theta_{\rm ld}}}$ and $U^{\boldsymbol{\theta_{\rm rd}}}$ are placed on the left and right of the encoded gate, combined with the identity gate on the remaining  $n-a$ qubits. Note that these parameterized unitary operators are independent of the parameter $x$ of the original quantum gate here. 

\begin{figure}[ht]
% \vskip 0.2in
% \begin{center}
\centering
\centerline{\includegraphics[width=\columnwidth]{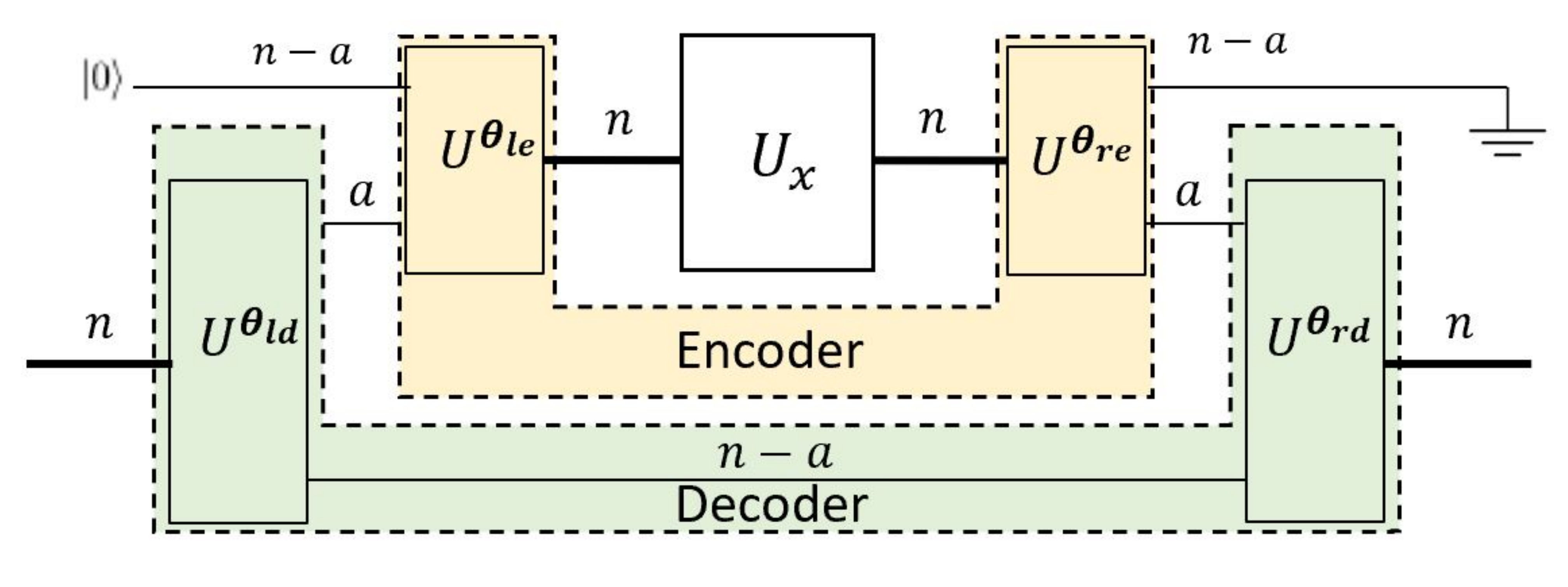}}
\caption{Implementation details of QAEGate. The parameterized unitary operators in yellow constitute the encoder while the parameterized unitary operators in green constitute the decoder. In the encoding phase, we insert the gate $U_x$ in the middle of the encoder, initialize $n-a$ input qubits to a fixed state $|0\>$, and discard $n-a$ output qubits to obtain a quantum channel acting on $a$ qubits. In the decoding phase, the decoder maps the channel back to an approximation of the initial $n$-qubit quantum gate.}
\label{fig:implementation}
% \end{center}
% \vskip -0.2in
\end{figure}

The choice of the parameters  $\boldsymbol{\theta_{\rm le}}$,  $\boldsymbol{\theta_{\rm re}}$,  $\boldsymbol{\theta_{\rm ld}}$, and $\boldsymbol{\theta_{\rm rd}}$ is optimized in order to maximize the similarity between the decoded quantum gate and the original quantum gate. As a similarity measure, we use the overlap between the Choi operators \cite{choi1975completely,jamiolkowski1972linear}, due to the relative ease of evaluating this quantity in numerical experiments.    As the optimization procedure, we will use  a stochastic gradient descent method, described in Subsection \ref{subsec:training}.

%========================================================================
%========================================================================
%\subsection{Implementation and Optimization}

To run the optimization, one needs first to fix the parametrization of the gates.  The naive  choice would be to pick a parametrization  that can represent arbitrary quantum gates. However, this approach has obvious  difficulties: 
\begin{itemize}[leftmargin=*]
\item  Describing an arbitrary $n$-qubit unitary transformation requires an exponential number of  parameters, and therefore a full optimization of the parameters in the autoencoder is only feasible for small values of $n$. 
\item Due to hardware restrictions of the current quantum computers,  it is difficult to implement arbitrary $n$-qubit unitary operators. This is especially true for the client, who may only have the ability to implement a small number of basic quantum gates. 
\end{itemize}
To address these problems,  we  propose an  implementation of QAEGate based on an approximation of  generic $n$-qubit unitary transformations that uses a  polynomial number of  parameters. We implement the parameterized unitary operators for $n$ qubits by decomposing them into a sequence of parameterized two-qubit unitary operators, where each pair of qubits corresponds to one such operator. \fig{unitary} presents the case when $n=4$. For each two-qubit parameterized unitary operator, it is further decomposed into a sequence of basic quantum gates, as depicted in \fig{2qubitunitary}.

\begin{figure}[ht]
% \vskip 0.2in
% \begin{center}
\centering
\centerline{\includegraphics[width=\columnwidth]{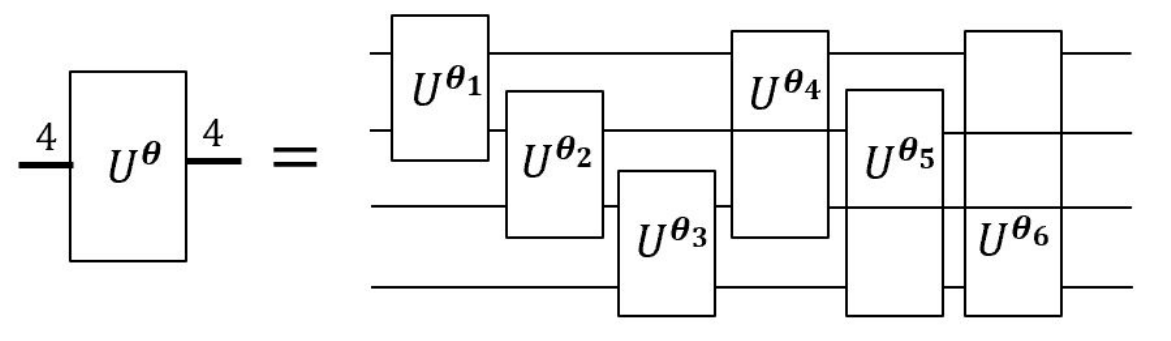}}
\caption{Decomposition of a parameterized 4-qubit unitary gate into 6 parameterized two-qubit gates.}
\label{fig:unitary}
% \end{center}
% \vskip -0.2in
\end{figure}

\begin{figure}[ht]
% \vskip 0.2in
% \begin{center}
\centering
\centerline{\includegraphics[width=\columnwidth]{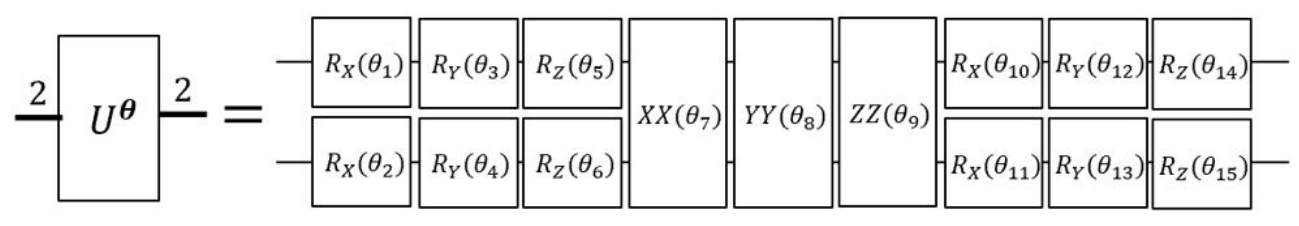}}
\caption{Decomposition of a two-qubit parameterized unitary gate.}
\label{fig:2qubitunitary}
% \end{center}
% \vskip -0.2in
\end{figure}

Here $R_X(\theta)$ is the quantum gate defined as follows. 
\begin{align*}
R_X(\theta)
:=\exp\left(i\theta \sigma_x
%   \begin{array}{cc}
%     0 & 1 \\
%     1 & 0 \\
%   \end{array}
% \right)
\right)
%\nonumber\\
=     \left(
  \begin{array}{cc}
    \cos\theta & -i\sin\theta\\
     -i\sin\theta & \cos\theta\\
  \end{array}
\right)
\label{eqn:rx_gate}
\end{align*}

The other two gates $R_Y(\theta)$ and $R_Z(\theta)$ are similar to $R_X(\theta)$ but for Pauli matrices $\sigma_y$ and $\sigma_z$.

The two-qubit gate $XX(\theta)$ is the Ising coupling gate that commonly appears in real-world quantum computers with trapped-ion architecture \cite{debnath2016demonstration}. It is defined as $XX(\theta):=\exp(i\theta \sigma_x\otimes \sigma_x)$ for the Pauli matrix $\sigma_x$, and written in matrix form as:
\begin{align*}
XX(\theta)
=     \left[
  \begin{array}{cccc}
    \cos\theta &0&0& -i\sin\theta\\
     0&\cos\theta&-i\sin\theta &0\\
     0&-i\sin\theta&\cos\theta &0\\
     -i\sin\theta&0&0& \cos\theta
  \end{array}
\right].
\end{align*}
$YY(\theta)$ and $ZZ(\theta)$ are Ising coupling gates corresponding to Pauli matrices $\sigma_y$ and $\sigma_z$ respectively.
\subsection{Training}\label{subsec:training}
To find the optimal parameters of QAEGate, we adopt classical gradient-based methods such as stochastic gradient descent \cite{kiefer1952stochastic}. 

We denote the parameters of QAEGate to be trained as $\boldsymbol\theta$. The training set is composed of $k$ quantum gates sampled randomly from $\{U_x\}$. In every iteration, we select a quantum gate $U$ from these $k$ gates randomly and insert it into the QAEGate. Then we calculate the overlap $f$ between the Choi state $C$ of $U$ and the Choi state $C'$ of the decoded state from the QAEGate. We define the loss function as $\mathcal{L} = 1 -f$ here and update the parameters of QAEGate according to the gradients $\nabla_{\boldsymbol{\theta}} \mathcal{L}$ until the decoded quantum gate is close enough to the original quantum gate or the maximum number of iterations is reached. We describe the complete training process in \algo{training_QAE}.

\begin{algorithm}
\DontPrintSemicolon
\KwData{$k$ quantum gates sampled randomly from $\{U_x\}$, maximum number of iterations $K$, learning rate $\eta$, threshold $\delta$.}
Initialize QAEGate parameters $\boldsymbol{\theta}$ randomly, $i = 0$, $f = 1$\;
\While{$i<K$ or $f>\delta$ }{Select the quantum gate $U$ randomly from $k$ input gates and insert it into the QAEGate \;
Calculate the Choi state $C$ 
% \bg{The calligraphic $C$ is usually for channels instead of Choi operators. Use Italic $C$ for Choi operators.}
of $U$\;
Get the decoded state from the QAEGate and calculate its corresponding Choi state $C'$\;
$f = {\rm Tr}(C C')$, which is the overlap between $C$ and $C'$ \;
$\mathcal{L} = 1 - f$ \;
Calculate $\nabla_{\boldsymbol{\theta}} \mathcal{L}$ and update $\boldsymbol{\theta}$ as $\boldsymbol{\theta} = \boldsymbol{\theta}-\eta \nabla_{\boldsymbol{\theta}} \mathcal{L}$ \;
$i = i + 1$\;}
\caption{Quantum autoencoders for quantum gates (QAEGate).\label{algo:training_QAE}}
\end{algorithm}

\subsection{Convergence Analysis}
We further prove that our proposed QAEGate can be efficiently trained by the stochastic gradient-based method in \algo{training_QAE}. Specifically, we obtain the following convergence guarantee:

\begin{theorem}[Convergence guarantee]\label{thm:theo1}
If we perform SGD, as in \algo{training_QAE}, to optimize $\boldsymbol\theta$ in the training of the QAEGate model, then the convergence rate is $T = \mathcal{O}( \frac{4\dim{ \boldsymbol\theta} }{\epsilon^4})$ for achieving the condition $\mathbb{E}[\Vert \nabla_{\boldsymbol{\theta}} \mathcal{L} \Vert^2]\leq \epsilon^2$.
\end{theorem}

The proof is provided in  the Appendix.

\section{Applications} \label{sec:app}
In this section we first apply  QAGate  to a basic scenario of the quantum cloud computing, in which we just consider single-round communication between the client and the server. Also, there is just one family of $n$-qubit quantum gates $\{U_x\}_{x\in \set X}$ stored in the server. 
% {\color{blue} [briefly describe]} 
Then, we extend the structure of QAEGate to apply it to two more general scenarios. One allows multiple-round communication between the client and the server. The other enables the server to store different families of $n$-qubit quantum gates.
% {\color{blue} [briefly describe]}.
\subsection{Basic scenario} 

Let us start from   the basic scenario introduced in \fig{task1} of  \sec{task}. The client starts from a generic $n$-qubit input $\rho$  and converts it into an $a$-qubit  state,  by applying the gate $ U^{\boldsymbol{\theta_{\rm ld}}}$ and storing aside $n-a$ qubits in a quantum memory. Then, the $a$ qubits are sent  to the server through a quantum communication link.    The server implements the encoder supermap, converting the original $n$-qubit gate $U_x$ into a quantum channel acting on $a$ qubits.  The encoded channel is applied to an  $a$-qubit state received from the client, and produces an   $a$-qubit state is sent back to the client through the quantum communication link.   Finally, the client performs the gate $ U^{\boldsymbol{\theta_{\rm rd}}}$ on the $a$ qubits received from the server and on the $n-a$ qubits previously stored in the quantum memory. Overall, the operations performed at the client's end implement the decoder, transforming the $a$-qubit channel implemented by the server into an approximation of the target $n$-qubit gate.  

An advantage of the structure of QAEGate is that it can fulfill the server's confidentiality. Specifically, the server wants to forbid the client from obtaining the detailed implementation of the gates, which is described by $x$ in this scenario. For sure, some information about $x$ must be leaked to the client since the client will be approximately granted one access to $U_x$, which could be used to extract information about $x$, but we claim that an arbitrary malicious client could not obtain more information than this minimal amount. Formally, we obtain the following confidentiality guarantee for our model:

\begin{theorem}[Confidentiality guarantee]\label{thm:theo2}
In this protocol, the information about $x$ obtained by an arbitrary malicious client is no larger than the information about $x$ obtainable by accessing $U_x$ once.
\end{theorem}

\begin{proof}

We first consider how much information about $x$ can be extracted given one copy of $U_x$. We assume $x$ follows a certain distribution, and we denote the random variable of $x$ as $X$. The most general way to extract information from a gate $U_x$ is to insert it into a quantum circuit and make measurements on the output of the circuit, as shown in \fig{appendix_conf}. Let the output of the circuit to be $Y$, which is a random variable correlated with $X$. We write $Y=\mathcal{F}(\mathcal{U}_X)$ to denote the relationship between $X$ and $Y$. The information one can know about $X$, by obtaining the value of $Y$, is the mutual information $I(X;Y)$. The maximum amount of information one can learn about $X$ is obtained by optimizing the circuit, which is $I_{\max}:=\max_{\mathcal{F}} I(X;Y)$. Since any party who is able to access $U_x$ could obtain this amount of information, the server, who grants the client one use of $U_x$, would inevitably leak $I_{\max}$ amount of information about $x$ to the client.

\begin{figure} \centering
\subfigure[A general quantum circuit $\mathcal{F}$ for extracting information about $x$ from $U_x$.] {
 \label{fig:appendix_conf}
\includegraphics[width=50mm]{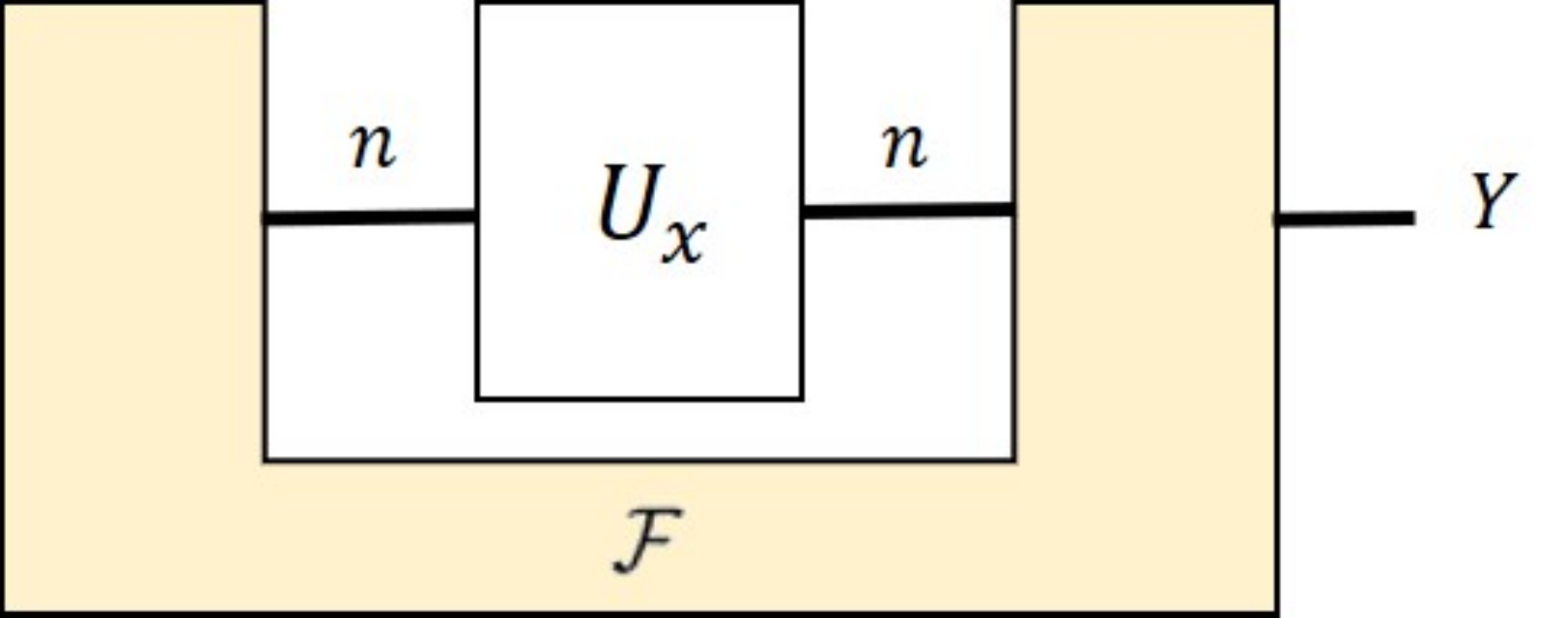} 
}
\hspace{.3in}
\subfigure[A general quantum circuit utilized by a malicious client to extract information about $x$ by communicating with the server.]
% , especially for LSN circuit, albeit composed of $1034$ simple quantum gates. ]
{
\label{fig:appendix_conf2}
\includegraphics[width=70mm]{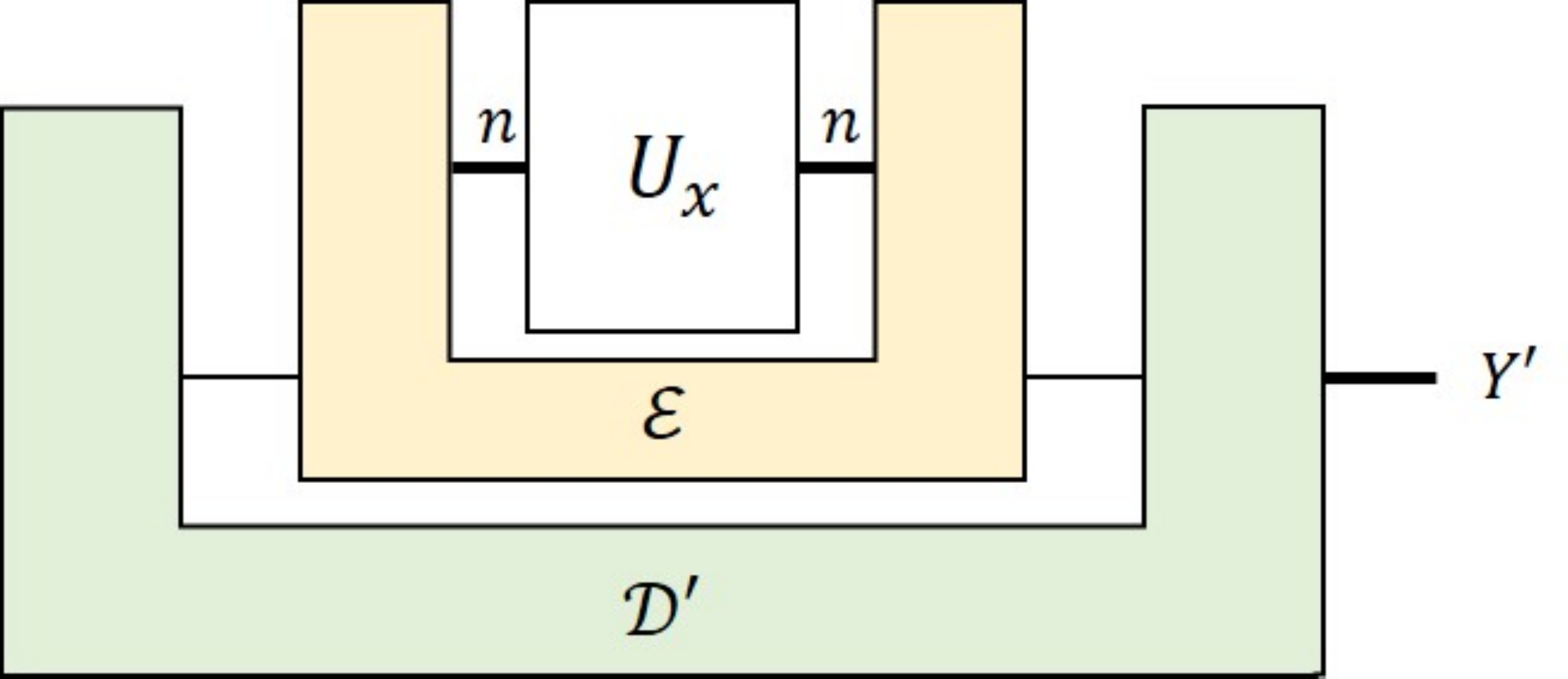}
}
\caption{}
\label{fig:appendix}
\end{figure}

Next, we consider a malicious client who wants to obtain the most information about $x$. The server interacts with the client via the encoder, which is equivalent to that the client has access to the encoded gate, which is a  channel $\mathcal{E}(\mathcal{U}_x)$. The most general way for the client to extract information about $x$ is to insert  $\mathcal{E}(\mathcal{U}_x)$ into a circuit and measure the output of the circuit, as shown in \fig{appendix_conf2}. Let this circuit be $\mathcal{D}'$ and its output be $Y'$, which is related to $X$ by $Y'=(\mathcal{D}'\circ\mathcal{E})(\mathcal{U}_X)$. The maximum information the client could obtain is $\max_{\mathcal{D}'}I(X;Y')$.

Now, since $\mathcal{E}$ is determined by the server, the circuit $\mathcal{D}'\circ\mathcal{E}$ with variable $\mathcal{D}'$ forms a subset of the range of $\mathcal{F}$, which contains all possible circuits taking $U_x$ as input. Therefore, the optimization over $\mathcal{D}'$ is upper bounded by the optimization over $\mathcal{F}$, and thus we have
\begin{align}
    \max_{\mathcal{D}'}I(X;Y') \leq I_{\max} \,.
\end{align}
The equation above proves \thm{theo2}, since the left hand side is the maximal information obtainable by an arbitrary malicious client, and the right hand side is the inevitable information leakage given one access to $U_x$.

\end{proof}

\subsection{Multi-round communication scenario}
In this scenario, we relax the assumption of one-round communication and the client can communicate with the server for multiple rounds.
In each round, the client can send to and receive from the sever an $a$-qubit quantum state  through the quantum communication link. For this multi-round communication scenario, the server can exploit $U_x$ more than once. Here we depict the case of two-round communication in \fig{task2}. 
\begin{figure}[ht]
% \vskip 0.2in
% \begin{center}
\centering
\centerline{\includegraphics[width=80mm]{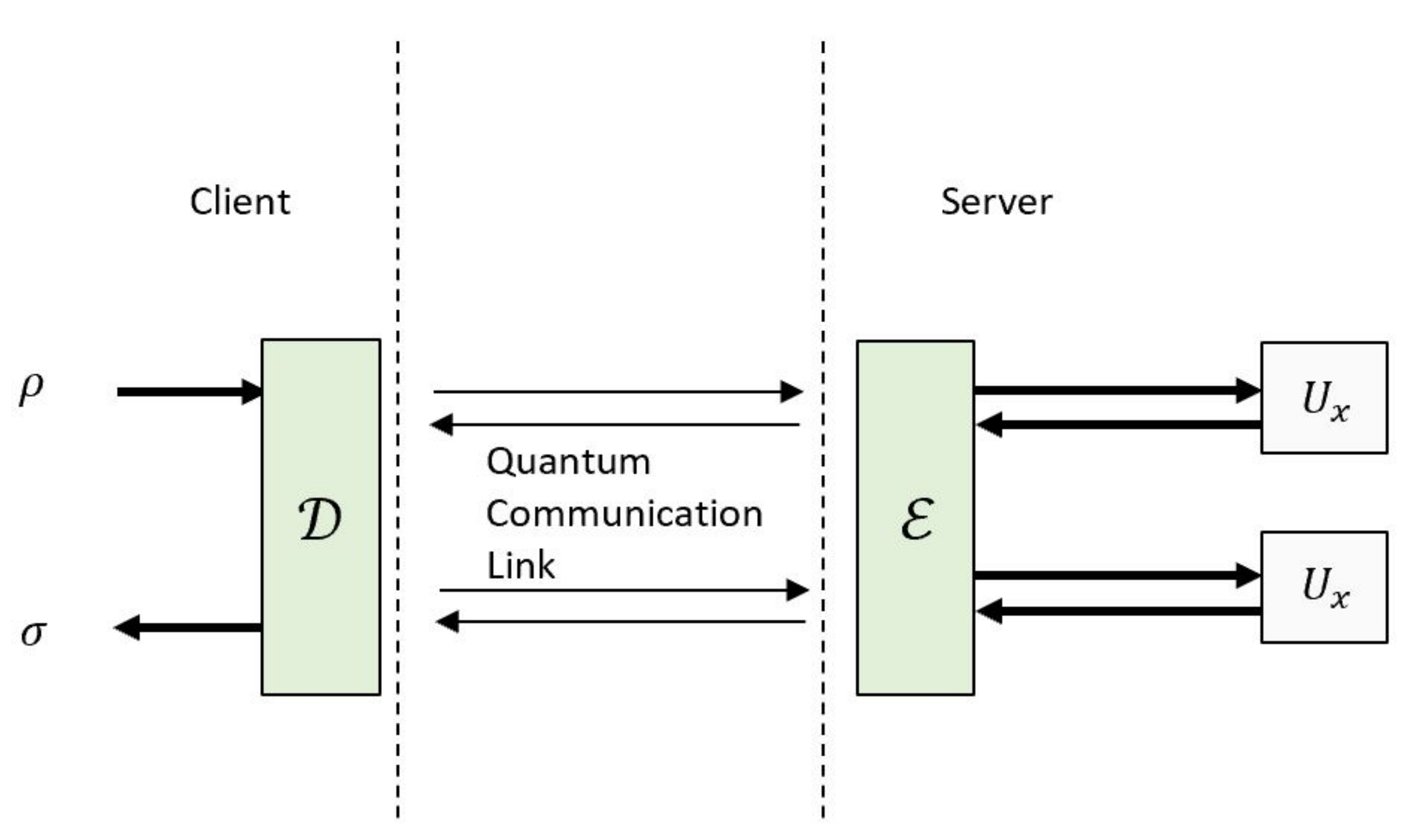}}
\caption{The multi-round communication scenario. The client can communicate with the server for multiple rounds.}
\label{fig:task2}
% \end{center}
% \vskip -0.2in
\end{figure}

For the optimal performance, a general protocol may utilize different encoders and decoders for each round of communication. Here we keep the remaining $n-a$ qubits of the output of the encoder in the first round and regard it as a part of the input of the encoder in the second round. We depict such a variation of the structure of the QAEGate model in \fig{struc2}. The encoders and decoders of the first and second round are separately parameterized.
\begin{figure}[ht]
% \vskip 0.2in
% \begin{center}
\centering
\centerline{\includegraphics[width=\columnwidth]{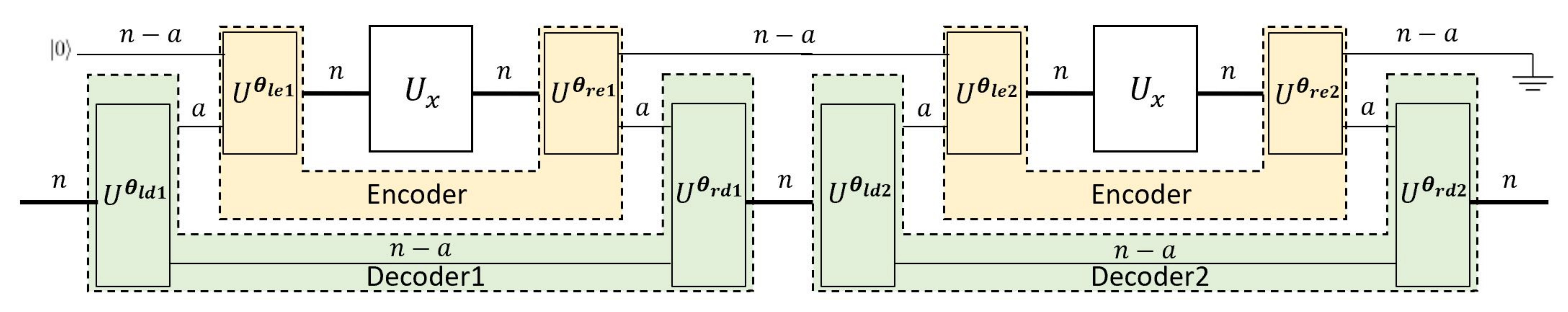}}
\caption{QAEGate for multi-round communication scenario.}
\label{fig:struc2}
% \end{center}
% \vskip -0.2in
\end{figure}

\subsection{Sequence of gates scenario}
In this scenario, we consider a more general situation that the client expects the server to execute a sequence of gates from different parametric families. In \fig{task3}, we depict the simplest case that the length of the sequence is $2$. Here there are two classes of quantum gates, $\{U_{x}\}_{x\in \set X}$ and $\{U_{y}\}_{y\in \set Y}$, stored in the server.  The client wants to obtain $\map{U}_{y}(\map{U}_{x}(\rho))$ by communicating with the server.

\begin{figure}[ht]
% \vskip 0.2in
% \begin{center}
\centering
\centerline{\includegraphics[width=80mm]{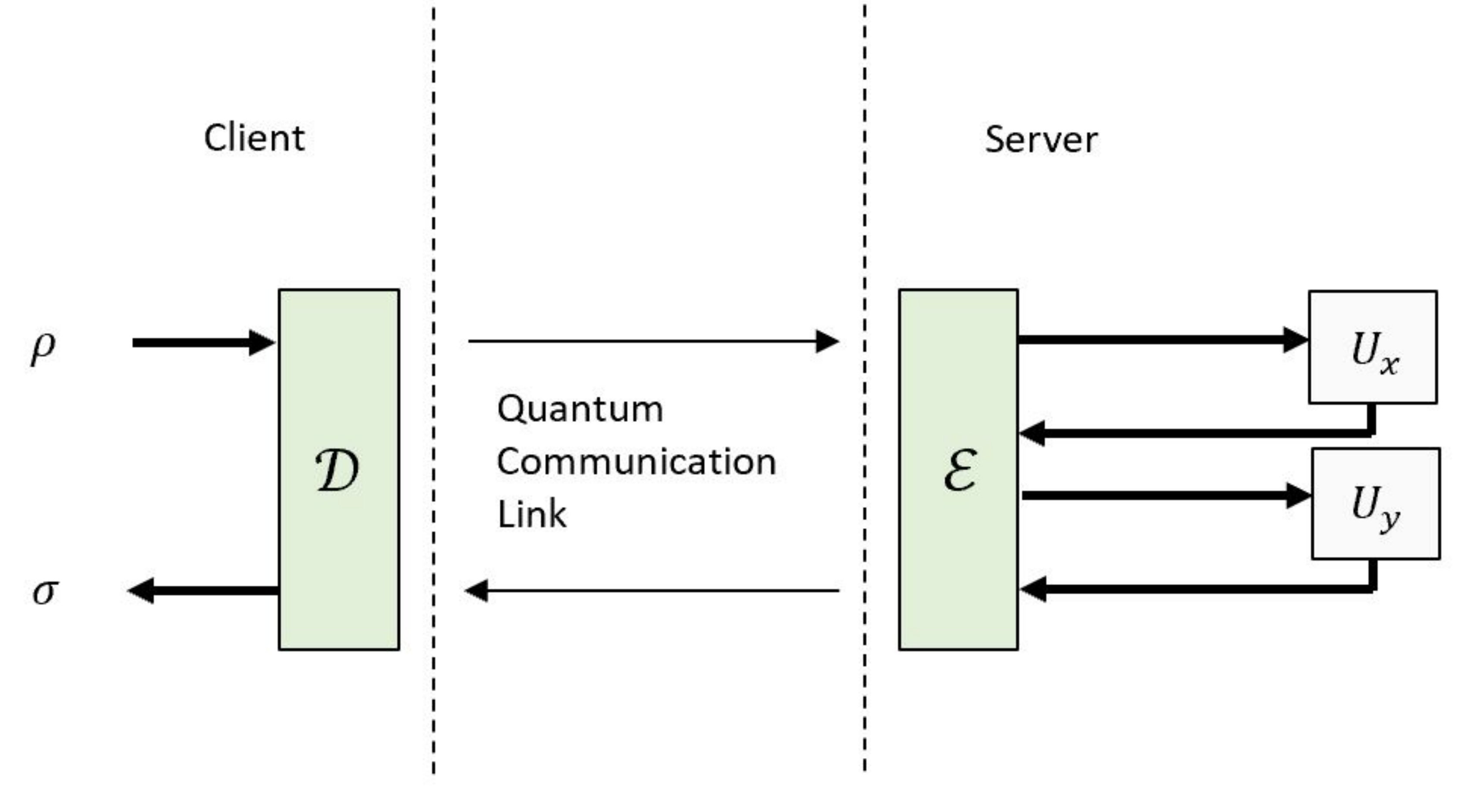}}
\caption{Sequence of gates scenario. The client expects the server to execute a sequence of gates from different parametric families.}
\label{fig:task3}
% \end{center}
% \vskip -0.2in
\end{figure}

In order to handle this scenario, we design the encoder and the decoder with different structures and depict the modified QAEGate model in \fig{struc3}. 

\begin{figure}[ht]
% \vskip 0.2in
% \begin{center}
\centering
\centerline{\includegraphics[width=\columnwidth]{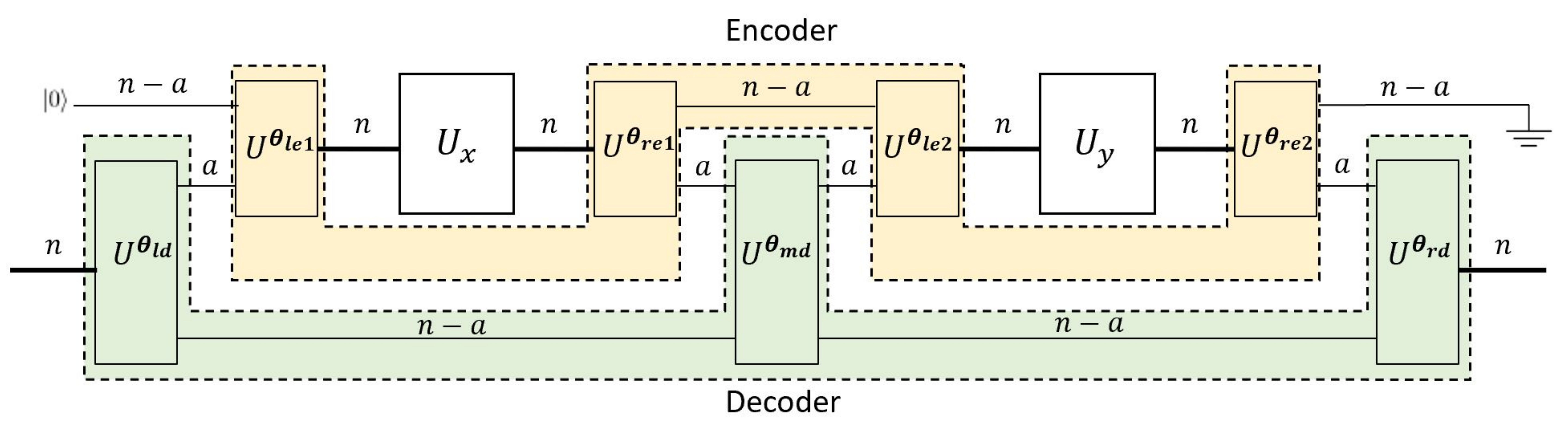}}
\caption{QAEGate for sequence of gates scenario.}
\label{fig:struc3}
% \end{center}
% \vskip -0.2in
\end{figure}

The encoder is composed of four parameterized unitary operators, denoted by $U^{\boldsymbol{\theta_{le1}}}$, $U^{\boldsymbol{\theta_{re1}}}$, $U^{\boldsymbol{\theta_{le2}}}$ and $U^{\boldsymbol{\theta_{re2}}}$. The first two operators are placed on both sides of $U_x$ while the other two are placed on both sides of $U_y$. The decoder is composed of three parameterized parameterized unitary operators, denoted by $U^{\boldsymbol{\theta_{\rm ld}}}$, $U^{\boldsymbol{\theta_{md}}}$ and $U^{\boldsymbol{\theta_{\rm rd}}}$.

\section{Numerical Experiments}\label{sec:exper}
In this section, we examine the performance of our proposed models by numerical simulation. We show the effectiveness of our proposed models in different settings for gates based on the Heisenberg model. All simulations are implemented by Tensorflow quantum \cite{broughton2020tensorflow} and performed on a single GPU. We train the models by stochastic gradient descent in all of the experiments.

\subsection{Heisenberg model}
%We mainly discuss the class of Heisenberg gates in our experiments.
Heisenberg model \cite{baxter2016exactly} is a famous statistical mechanical model, which describes a magnetic system of half spins. It can be defined by the Hamiltonian
\begin{align*}
\tilde{H} =& -\frac{1}{2}\sum_{j=1}^{n-1}(J_{\rm x}\sigma_j^x\sigma_{j+1}^x+J_{\rm y}\sigma_j^y\sigma_{j+1}^y+J_{\rm z}\sigma_j^z\sigma_{j+1}^z) \\ & -\frac12 \sum_{j=1}^n h\sigma_j^z,
\end{align*}
where $J_{\rm x}$, $J_{\rm y}$, $J_{\rm z}$ and $h$ are constants representing the strength of the coupling and the external magnetic field, respectively.
$\sigma^x$, $\sigma^y$, $\sigma^z$ are the Pauli matrices. $n$ is the number of qubits in the quantum system. We define Heisenberg gates as the quantum gates which represent evolution defined by a Hamiltonian operator of Heisenberg model. These Heisenberg gates can be represented by
\begin{align*}
U(t) = \exp(-i\tilde{H} t),
\end{align*}
where $\tilde{H}$ is the Hamiltonian of a Heisenberg model and $t$ is the evolution time.

\subsection{Simulation results}
\paragraph{Basic Scenario.} 
In this scenario, we set $\{U_{x}\}$ stored in the sever as a class of Heisenberg gates with $J_{\rm x} = J_{\rm y} = J_{\rm z} = 0.1$ and $h = 0.5$. 

Here we define the evolution time $t$ as the parameter $x$, unknown to the client. The client can only transmit $a=1$ qubit or $a=2$ qubits through the quantum communication link. We consider three cases, $n=2,3,4$, in the experiment. We exploit a training set composed of $50$ gates to optimize the compression model and validate it on the test set composed of other $10$ gates in each experiment.

We show the training curves of the QAEGate model in \fig{s1}. Here the x-axis represents the epoch number of the training, and the y-axis represents the overlap between the Choi states of the original gate and the decoded gate. We can see that all of the training processes converge finally, which conforms to our convergence analysis in \thm{theo1}. 

\begin{figure}[ht]
% \vskip 0.2in
% \begin{center}
\centering
\centerline{\includegraphics[width=80mm]{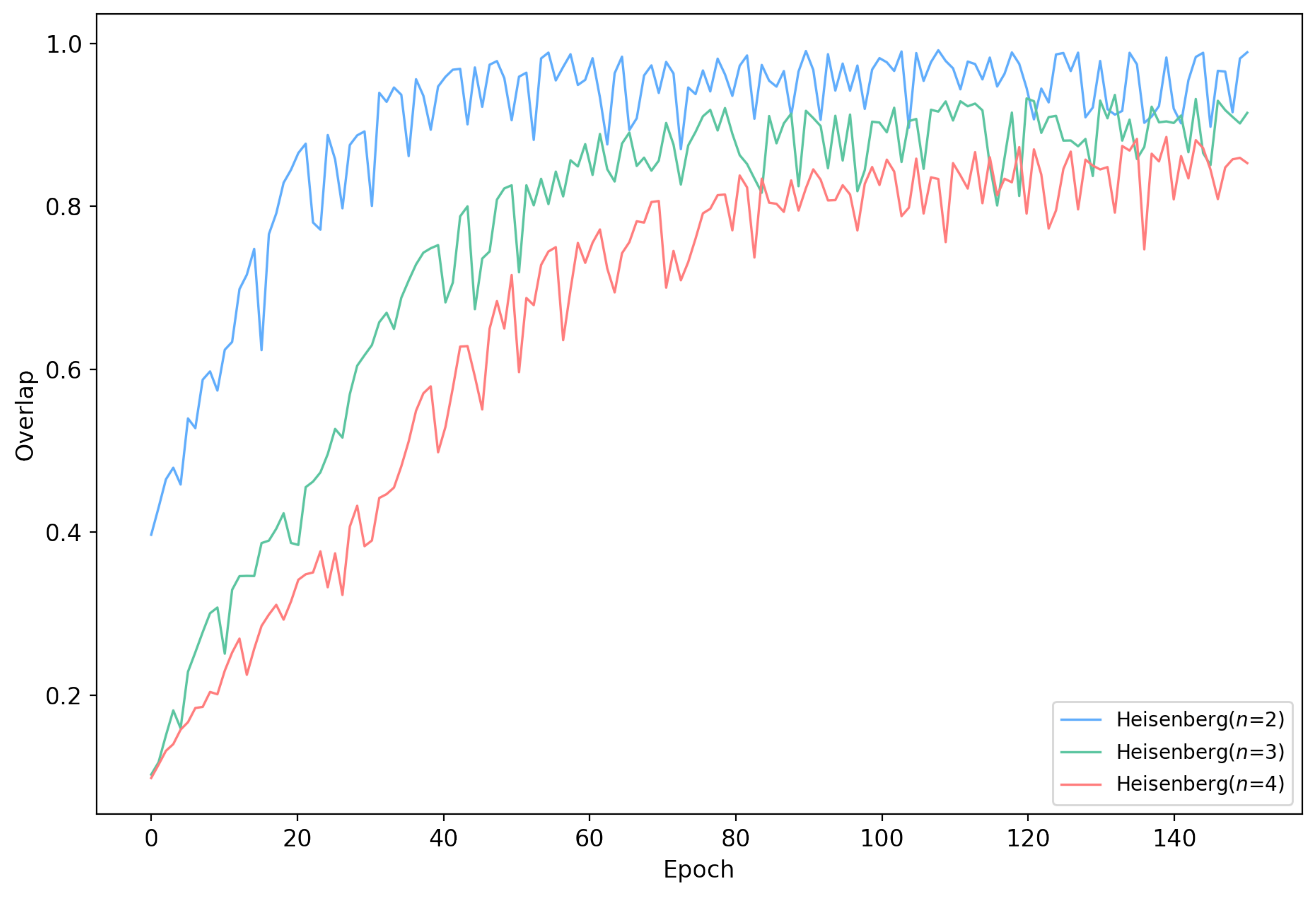}}
\caption{Training curves of the QAEGate for the basic scenario. We train three QAEGate models for three classes of Heisenberg gates, $n=2,3,4$, respectively.  All of training curves converge finally, which conforms to our convergence analysis.}
\label{fig:s1}
% \end{center}
\vskip -0.1in
\end{figure}

Furthermore, we present the performance of our QAEGate model on test sets in \fig{s2_1}.  The y-axis represents the average overlap values between the Choi states of the decoded quantum gate and the original quantum gate.  We can see that our QAEGate model achieves satisfactory performance in all of the experiments, which confirms the effectiveness of our proposed communication model empirically. However, we can observe that QAEGate performs better when the dimension of the class of Heisenberg gates $\{U_{x}\}$ is small. 

\begin{figure}[ht]
% \vskip 0.2in
% \begin{center}
\centering
\centerline{\includegraphics[width=80mm]{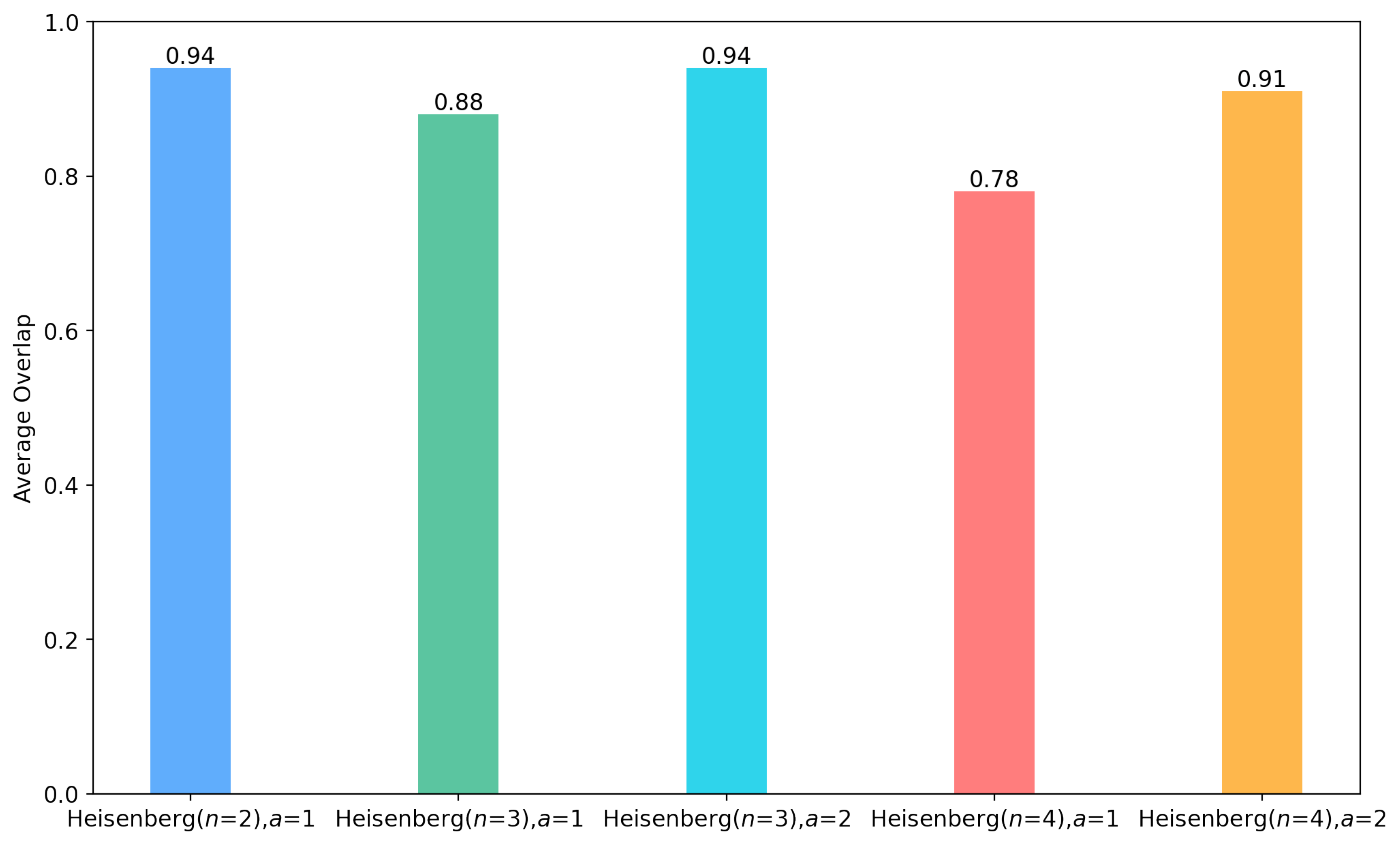}}
\caption{The performance of our QAEGate model on test sets of the basic scenario.}
\label{fig:s2_1}
% \end{center}
\vskip -0.1in
\end{figure}

\paragraph{Multi-round communication scenario.}
In this scenario, we allow two-round communication between the server and the client while they can only transmit $a=1$ qubit through the quantum communication link at a time. 

We present the performance of our QAEGate model on test sets of the two-round communication scenario in \fig{s2_2}. The results indicate that our proposed model performs better when additional round of communication is available, compared with the basic scenario with $a=1$. Especially for the case of $n=4$, the average overlap value increases by $0.12$ compared with the basic single-round scenario. We also compare the case of $a=2$ in the basic scenario and the case of $a=1$ in the two-rounds communication scenario. In both cases, two qubits are sent to and received from the server in total. We find that our proposed models have similar performance. 

\begin{figure}[ht]
% \vskip 0.2in
% \begin{center}
\centering
\centerline{\includegraphics[width=80mm]{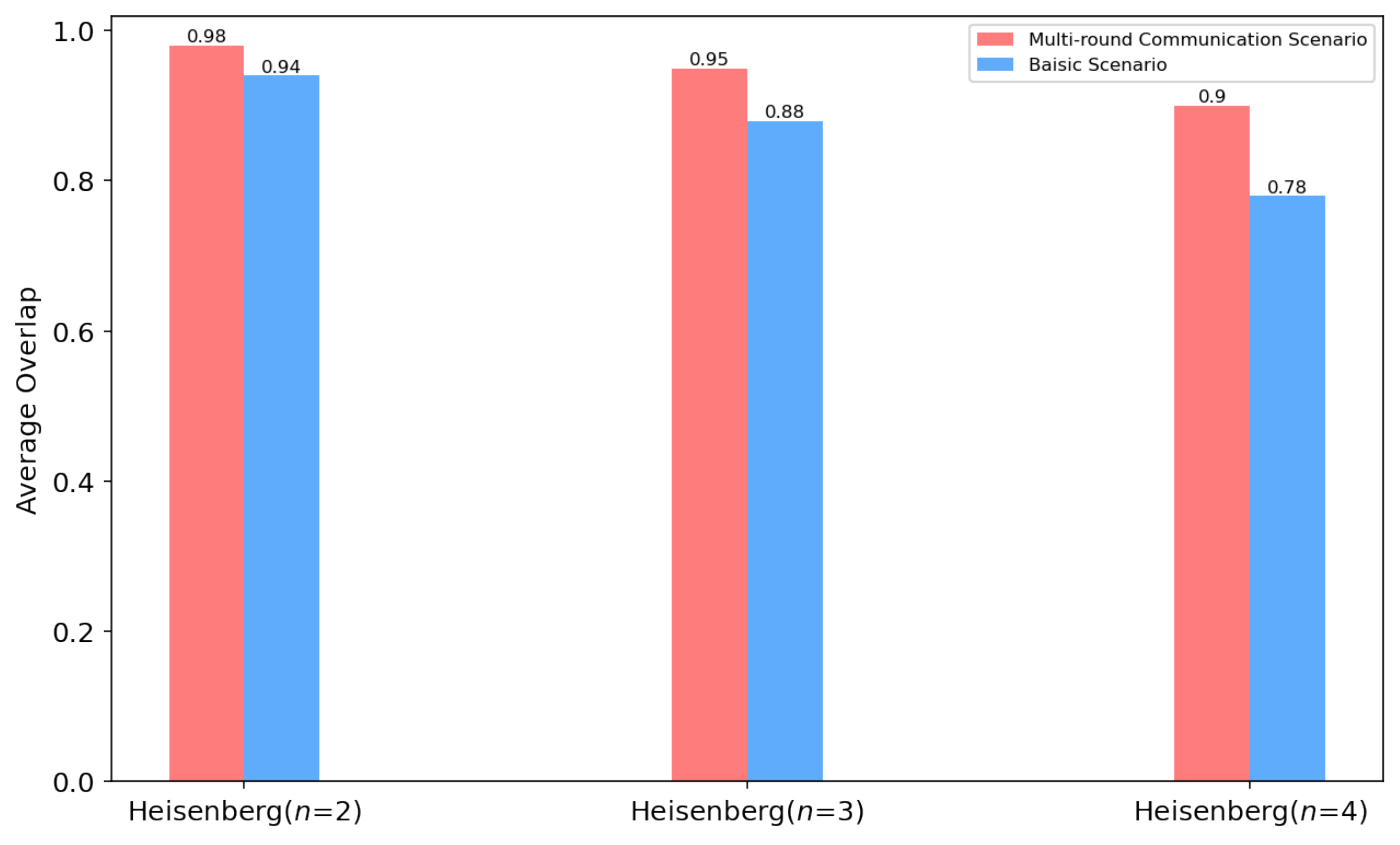}}
\caption{The performance of our QAEGate model on test sets of the multi-round communication scenario. }
\label{fig:s2_2}
% \end{center}
\vskip -0.1in
\end{figure}

\paragraph{Sequence of gates scenario.}
In this scenario, we set $\{U_{x}\}$ stored in the sever as a class of Heisenberg gates with $J_{\rm x} = J_{\rm y} = J_{\rm z} = 0.1$ and $h = 0.5$ and $\{U_{y}\}$ as a class of Heisenberg gates with $J_{\rm x} = J_{\rm y} = 0.1$ and $J_{\rm z} = h = 0.5$. Here We define the evolution times of them as $x$ and $y$ whose values can be different.  The client still can only transmit $a=1$ qubit through the quantum communication link. We also consider three cases, $n=2,3,4$, in the experiment. The simulation results on test sets of this scenario are shown in \fig{s2_3} and our proposed model achieved satisfactory performance in all three cases.

\begin{figure}[ht]
% \vskip 0.2in
% \begin{center}
\centering
\centerline{\includegraphics[width=80mm]{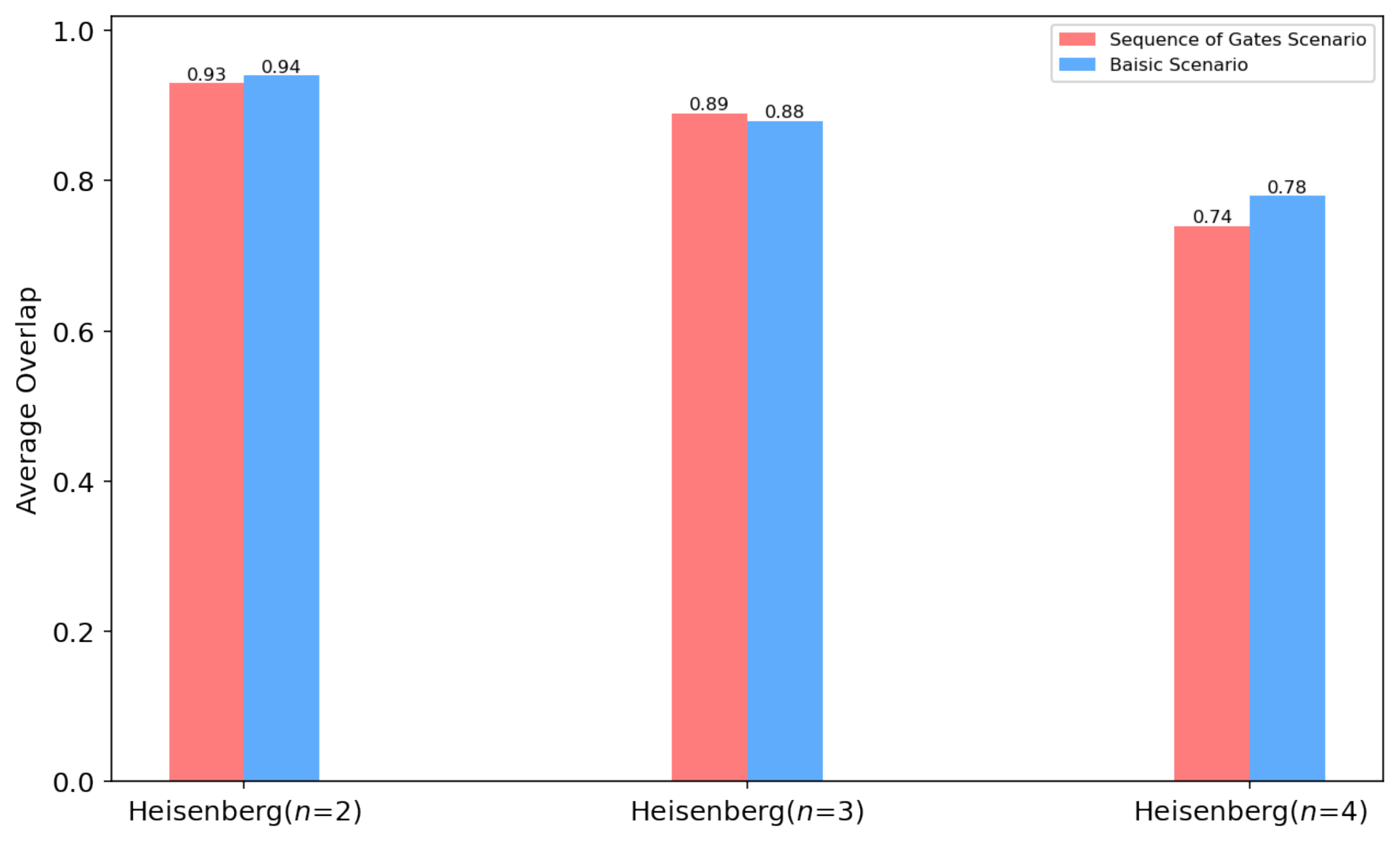}}
\caption{The performance of our QAEGate model on test sets of the sequence of gates scenario. }
\label{fig:s2_3}
% \end{center}
\vskip -0.1in
\end{figure}

\section{Discussion and conclusions}\label{sec:conclu}
We proposed a quantum machine learning model, QAEGate, for communication-efficient quantum cloud computing, and applied it to single-round and multi-round communication scenarios. In the single-round scenario, the protocol avoids leakage of information about the server's computation and achieves a satisfactory performance in terms of overlap with the target gate versus the number of qubits in its compressed version.  By adjusting the structure of QAEGate, we applied our method to other two scenarios, involving multiple gates and multiple rounds of communication between the client and server. 

The performance of our method has been tested numerically in a number of  examples. Due to the exponential complexity of simulating quantum gates on classical computers,  such analysis is necessarily limited to  scenarios involving a small number of qubits. On the other hand, our method is meant to be eventually run  on real quantum computers, on which the complexity of QAEGate grows only polynomially. 
  While the algorithm The current progress of quantum hardware suggests that a real-world test may take place not too far in the future. 

An application of the techniques developed in this paper is learning unknown quantum gates \cite{bisio2010optimal,mo2019quantum,sedlak2020probabilistic}. In this task, the goal is to imprint an unknown gate into the state of a quantum memory from which the gate can be accessed at a later time. Our QAEGate can be adapted to this task by adjusting the structure of the encoder and the decoder so that quantum gates can be encoded into quantum states.    Referring to Figure \ref{fig:implementation},  the idea is to  choose different number of qubits $a$ and $a'$ for the input and output wires of the encoder, respectively. In particular, one can set $a$ to zero, meaning that the encoder produces a quantum state of $a'$ qubits. The number of qubits $a'$ can be regarded as the size of the quantum memory of a learning machine designed to learn a generic gate from a given parametric family of quantum gates.  An interesting feature of our method is that it allows one to set the memory size ({\em i.e.} to fix $a'$) and to achieve approximate learning subject to constraints on the available quantum memory.

Another interesting direction of future research is  assessing  the power of quantum processors \cite{bishop2017quantum,moll2018quantum}.  It has been often pointed out that the number of qubits in a quantum processor is not an appropriate measure, because the presence of noise generally limits the set of achievable computations.  Hence, it is important to have more inclusive measures, that take into account not only the raw number of qubits, but also the size of the set of operations implementable on them.  The  approach of quantum gate compression  sheds light into this problem  by providing a rigorous notion of ``effective size'' of a quantum processor, defined as the minimum number of qubits on which the operations of the processor can be compressed with sufficiently high fidelity.  

More specifically,  the set of operations implemented by  a given quantum  processor  can be regarded as a parametric family of (generally noisy) quantum gates. When the processor is not universal, such a family may be compressible into operations acting on a smaller number of qubits, meaning that the computations achieved by the processor can be in principle simulated with a smaller quantum computer.  The size of this minimal simulator can then be regarded as a  measure of the computing power of the original quantum processor. Such a measure is  independent of the  hardware implementation, and could be in principle evaluated by connecting the given processor to a reference quantum computer that runs  the  QAEGate algorithm, or a generalization of it working with noisy input gates.     

A limitation  of the above approach is that the evaluation of the minimum number of qubits  requires  a trusted quantum computer of the same size of the original processor. On the other hand, once a sufficiently large universal quantum computer becomes available,  it can be used as a standard reference for assessing the power of other quantum processors, and to construct efficient simulations that can be run on universal quantum computers of smaller size.

%%%%%%%%%%%%%%%%%%%%%%%%%%%%%%%%%%%%%%%%%%%%%%%%%%%%%%%%%%%%%%%%%%%%%%%%%%%%%%

\bibliography{QAEGate}
\bibliographystyle{plain}
% \begin{thebibliography}{9}
% \bibitem{examplecitation}
%   Name Surname,
%   \href{https://doi.org/10.22331/
%         idonotexist}{Quantum
%         \textbf{123}, 123456 (1916).}

% \bibitem{biblatexsubmittingtothearxiv}
%   StackExchange discussion on \href{http://tex.stackexchange.com/questions/26990/biblatex-submitting-to-the-arxiv}{"Biblatex: submitting to the arXiv'' (2017-01-10)}

% \bibitem{arxivpdfoutput}
%   Help article published by the arXiv on \href{https://arxiv.org/help/submit_tex}{"Considerations for TeX Submissions'' (2017-01-10)}

% \bibitem{howtogetdoilinksinbibliography}
%   StackExchange discussion on \href{http://tex.stackexchange.com/questions/3802/how-to-get-doi-links-in-bibliography}{"How to get DOI links in bibliography'' (2016-11-18)}

% \bibitem{automaticallyaddingdoifieldstoahandmadebibliography}
%   StackExchange discussion on \href{http://tex.stackexchange.com/questions/6810/automatically-adding-doi-fields-to-a-hand-made-bibliography}{"Automatically adding DOI fields to a hand-made bibliography'' (2016-11-18)}
% \end{thebibliography}

\onecolumn\newpage
\appendix

\section{Proof of \thm{theo1}}
According to \cite{allen2017natasha}, we have the following lemma about convergence of SGD for nonconvex functions:
\begin{lemma}\label{lem:converge}
If a function $f(x)$ is $L_1$-smooth and we perform SGD update $x_{t+1}\leftarrow x_t - \eta \nabla f_i(x_t)$ each time for a random $i\in [n]$, then the convergence rate is $T = \mathcal{O}(\frac{L_1}{\epsilon^4})$ for finding $\mathbb{E}[\Vert \nabla f(x) \Vert^2]\leq \epsilon^2$.
\end{lemma}

Therefore, we first discuss the smoothness property of the loss function $\mathcal{L}(\boldsymbol\theta)$ in the training of QAEGate model. Note that the first- and second-order Lipschitz smoothness are equivalent to Lipschitz continuity of the first- and second-order derivatives. So we start with a lemma about the Lipschitz continuity of multivariate functions \cite{sweke2020stochastic}.

\begin{lemma}\label{lem:lip}
Given some function $f \colon \R^M \to \R$, if all partial derivatives of $f$ are continuous, then for any $a, b \in \R$ the function $f \colon [a, b]^M \to \R$ is $L$-Lipschitz continuous with
\begin{align}
	L = \sqrt{M}\max_{j \in \{1,\dots,M\} } \sup_{\boldsymbol{x}\in [a, b]^M} \left|\frac{\partial f(\boldsymbol{x})}{\partial x_j}\right|.
\end{align}
\end{lemma}

Next we prove a more general result about the smoothness of measurement probability of quantum circuits, and then show that $\mathcal{L}(\boldsymbol\theta)$ fits into this result.

\begin{lemma}\label{lem:lip_smooth}
Consider a quantum circuit consisting of any number of fixed unitary gates and $M$ variable unitary gates $U_1(\theta_1),\dots,U_M(\theta_M)$, where $U_j(\theta_j) := \exp(i\theta_j H_j)$ for some Hermitian operator $H_j$. Then the probability of any measurement outcome of the output of this circuit is $L_1$-smooth and $L_2$-second-order smooth with respect to $\boldsymbol{\theta}=(\theta_1,\dots,\theta_M)$, where
\begin{align}
	&L_1 = 4MH_{\max}^2\\
	&L_2 = 8M^{3/2}H_{\max}^3
\end{align}
where $H_{\max}:=\max_{j} \|H_j\|_2$ and $\|\cdot\|_2$ is the induced operator norm defined as $\|A\|_2 := \sup\frac{\|Ax\|_2}{\|x\|_2}$.
\end{lemma}
\begin{proof}
Let $|\psi_0\rangle$ be the initial state of the circuit. Without loss of generality, we assume the variable unitary gates are labeled as $U_1(\theta_1)$ to $U_M(\theta_M)$ according to the order they are applied. Then the output of this circuit can be written as $|\psi_{\boldsymbol{\theta}}\rangle := V_M U_M(\theta_M) \dots V_1 U_1(\theta_1) V_0 |\psi_0\rangle$, where $V_1,\dots,V_M$ denotes the fixed unitary gates. A measurement outcome can be described by a positive operator-valued measure (POVM) element $P$ with the property that $P$ and $I-P$ are both positive semidefinite, where $I$ is the identity operator of the output. The probability of this outcome is
\begin{align}
	f(\boldsymbol{\theta}) := \Tr[P |\psi_{\boldsymbol{\theta}}\rangle \langle\psi_{\boldsymbol{\theta}}|] = \langle\psi_{\boldsymbol{\theta}}| P |\psi_{\boldsymbol{\theta}}\rangle \,.
\end{align}

To study the smoothness of $f(\boldsymbol{\theta})$, we take the partial derivatives:
\begin{align*}
	\frac{\partial f(\boldsymbol{\theta})}{\partial \theta_j} & = \frac{\partial\langle\psi_{\boldsymbol{\theta}}|}{\partial \theta_j} P |\psi_{\boldsymbol{\theta}}\rangle + \langle\psi_{\boldsymbol{\theta}}| P \frac{\partial|\psi_{\boldsymbol{\theta}}\rangle}{\partial \theta_j} = 2 \Re \left( \frac{\partial\langle\psi_{\boldsymbol{\theta}}|}{\partial \theta_j} P |\psi_{\boldsymbol{\theta}}\rangle\right) \\
	\frac{\partial^2 f(\boldsymbol{\theta})}{\partial \theta_j \partial\theta_k} & = 2 \Re \left( \frac{\partial^2\langle\psi_{\boldsymbol{\theta}}|}{\partial \theta_j \partial\theta_k} P |\psi_{\boldsymbol{\theta}}\rangle + \frac{\partial\langle\psi_{\boldsymbol{\theta}}|}{\partial \theta_j} P \frac{\partial|\psi_{\boldsymbol{\theta}}\rangle}{\partial\theta_k} \right)\\
	\frac{\partial^3 f(\boldsymbol{\theta})}{\partial \theta_j \partial\theta_k \partial\theta_l} & = 2 \Re \left( \frac{\partial^3\langle\psi_{\boldsymbol{\theta}}|}{\partial \theta_j \partial\theta_k \partial\theta_l} P |\psi_{\boldsymbol{\theta}}\rangle +\frac{\partial^2\langle\psi_{\boldsymbol{\theta}}|}{\partial \theta_j \partial \theta_k} P \frac{\partial|\psi_{\boldsymbol{\theta}}\rangle}{\partial\theta_l}+ \frac{\partial^2 \langle\psi_{\boldsymbol{\theta}}|}{\partial \theta_j \partial \theta_l} P \frac{\partial|\psi_{\boldsymbol{\theta}}\rangle}{\partial\theta_k} + \frac{\partial \langle\psi_{\boldsymbol{\theta}}|}{\partial \theta_j} P \frac{\partial^2|\psi_{\boldsymbol{\theta}}\rangle}{\partial\theta_k \partial \theta_l} \right)
\end{align*}

These partial derivatives can be bounded with vector and operator norms. For a unitary operator $U$, $\|U\|_2=1$. For the POVM element $P$, $\|P\|_2\leq 1$. The second-order partial derivatives can be bounded as:
\begin{align}
	\left|\frac{\partial^2 f(\boldsymbol{\theta})}{\partial \theta_j \partial\theta_k}\right| &
		\leq 2 \left| \frac{\partial^2\langle\psi_{\boldsymbol{\theta}}|}{\partial \theta_j \partial\theta_k} P |\psi_{\boldsymbol{\theta}}\rangle
			+ \frac{\partial\langle\psi_{\boldsymbol{\theta}}|}{\partial \theta_j} P \frac{\partial|\psi_{\boldsymbol{\theta}}\rangle}{\partial\theta_k} \right| \nonumber \\
		& \leq 2\left\|\frac{\partial^2\langle\psi_{\boldsymbol{\theta}}|}{\partial \theta_j \partial\theta_k}\right\|_2 \|P\|_2 \||\psi_{\boldsymbol{\theta}}\rangle\|_2
			+ 2\left\|\frac{\partial\langle\psi_{\boldsymbol{\theta}}|}{\partial \theta_j}\right\|_2 \|P\|_2 \left\|\frac{\partial|\psi_{\boldsymbol{\theta}}\rangle}{\partial\theta_k}\right\|_2 \nonumber \\
		& \leq 2\max_{a,b}\left\|\frac{\partial^2|\psi_{\boldsymbol{\theta}}\rangle}{\partial \theta_a \partial\theta_b}\right\|_2  + 2 \max_a \left\|\frac{\partial|\psi_{\boldsymbol{\theta}}\rangle}{\partial \theta_a}\right\|_2^2 \label{eq:partial2}
\end{align}
Similarly, for the third-order partial derivatives,
\begin{align}
	\left|\frac{\partial^3 f(\boldsymbol{\theta})}{\partial \theta_j \partial\theta_k \partial\theta_l} \right |
	& \leq 2 \left| \frac{\partial^3\langle\psi_{\boldsymbol{\theta}}|}{\partial \theta_j \partial\theta_k \partial\theta_l} P |\psi_{\boldsymbol{\theta}}\rangle + \frac{\partial^2\langle\psi_{\boldsymbol{\theta}}|}{\partial \theta_j \partial \theta_k} P \frac{\partial|\psi_{\boldsymbol{\theta}}\rangle}{\partial\theta_l}+ \frac{\partial^2 \langle\psi_{\boldsymbol{\theta}}|}{\partial \theta_j \partial \theta_l} P \frac{\partial|\psi_{\boldsymbol{\theta}}\rangle}{\partial\theta_k} + \frac{\partial \langle\psi_{\boldsymbol{\theta}}|}{\partial \theta_j} P \frac{\partial^2|\psi_{\boldsymbol{\theta}}\rangle}{\partial\theta_k \partial \theta_l} \right| \nonumber \\
	& \leq 2\max_{a,b,c}\left\|\frac{\partial^3 | \psi_{\boldsymbol{\theta}}\rangle }{\partial \theta_a \partial\theta_b \partial\theta_c}\right\|_2 + 6 \left(\max_{a,b}\left\|\frac{\partial^2 |\psi_{\boldsymbol{\theta}}\rangle}{\partial \theta_a \partial\theta_b}\right\|_2 \right)\left( \max_{a}\left\|\frac{\partial|\psi_{\boldsymbol{\theta}}\rangle}{\partial \theta_a }\right\|_2 \right) \label{eq:partial3}
\end{align}

Therefore, to show that the second- and third-order partial derivatives are bounded, it suffices to show that $\max_{a}\left\|\frac{\partial|\psi_{\boldsymbol{\theta}}\rangle}{\partial \theta_a }\right\|_2, \max_{a,b}\left\|\frac{\partial^2 |\psi_{\boldsymbol{\theta}}\rangle}{\partial \theta_a \partial\theta_b}\right\|_2$ and $\max_{a,b,c}\left\|\frac{\partial^3 | \psi_{\boldsymbol{\theta}}\rangle }{\partial \theta_a \partial\theta_b \partial\theta_c}\right\|_2$ are all bounded.
We observe that $\frac{\partial U_j(\theta_j)}{\partial j} = \frac{\partial \exp(i\theta_j H_j)}{\partial j} = iH_j U_j(\theta_j)$, and we have

\begin{align}
\max_{a}\left\|\frac{\partial|\psi_{\boldsymbol{\theta}}\rangle}{\partial \theta_a }\right\|_2 & = \max_{a}\left\| V_M U_M(\theta_M) \dots V_a i  H_a U_a(\theta_a) \dots U_1(\theta_1)V_0 |\psi_0\rangle \right\|_2 \nonumber\\
	& \leq \max_{a} \left\| V_M U_M(\theta_M) \dots V_a \right\|_2 \|H_a\|_2 \left\| U_a(\theta_a) \dots U_1(\theta_1)V_0 |\psi_0\rangle \right\|_2 \nonumber \\
	& = \max_{a} \|H_a\|_2  = H_{\max}
\end{align}
With similar derivations, we can obtain $\max_{a,b}\left\|\frac{\partial^2 |\psi_{\boldsymbol{\theta}}\rangle}{\partial \theta_a \partial\theta_b}\right\|_2^2 \leq H_{\max}^2$ and  $\max_{a,b,c}\left\|\frac{\partial^3 | \psi_{\boldsymbol{\theta}}\rangle }{\partial \theta_a \partial\theta_b \partial\theta_c}\right\|_2 \leq H_{\max}^3$. According to Eqs. (\ref{eq:partial2}) and (\ref{eq:partial3}), we obtain
\begin{align}
\left|\frac{\partial^2 f(\boldsymbol{\theta})}{\partial \theta_j \partial\theta_k}\right| & \leq 4 H_{\max}^2 \\
\left|\frac{\partial^3 f(\boldsymbol{\theta})}{\partial \theta_j \partial\theta_k \partial\theta_l} \right | & \leq 8 H_{\max}^3
\end{align}

Consider $\frac{\partial f(\boldsymbol{\theta})}{\partial \theta_j}$ as a function of $\boldsymbol{\theta}$. Since its partial derivatives $\frac{\partial^2 f(\boldsymbol{\theta})}{\partial \theta_j \partial\theta_k}$ are bounded by $4 H_{\max}^2$, according to \lem{lip}, $\frac{\partial f(\boldsymbol{\theta})}{\partial \theta_j}$ is $L_1'$-continuous with $L_1'=4 \sqrt{M} H_{\max}^2$. Then
\begin{align}
	\| \nabla f(\boldsymbol\alpha) - \nabla f(\boldsymbol\beta) \|_2 &\leq \sqrt{M}\max_j \left|\frac{\partial f(\boldsymbol{\alpha})}{\partial \theta_j} - \frac{\partial f(\boldsymbol{\beta})}{\partial \theta_j} \right| \nonumber\\&\leq \sqrt{M} L_1' \|\alpha - \beta\|_2 = 4 M H_{\max}^2 \|\alpha - \beta\|_2
\end{align}
and therefore $f$ is $L_1$-smooth with $L_1 = 4 M H_{\max}^2$. Similarly, consider $\frac{\partial^2 f(\boldsymbol{\theta})}{\partial \theta_j \partial\theta_k}$ as a function of $\boldsymbol{\theta}$. Since its partial derivatives $\frac{\partial^3 f(\boldsymbol{\theta})}{\partial \theta_j \partial\theta_k \partial\theta_l}$ are bounded by $8 H_{\max}^3$, according to \lem{lip}, $\frac{\partial^2 f(\boldsymbol{\theta})}{\partial \theta_j \partial\theta_k}$ is $L_2'$-continuous with $L_2'=8\sqrt{M} H_{\max}^3$. Then
\begin{align}
	\| \nabla^2 f(\boldsymbol\alpha) - \nabla^2 f(\boldsymbol\beta) \|_2 &\leq M \max_{j,k} \left|\frac{\partial^2 f(\boldsymbol{\alpha})}{\partial \theta_j \partial\theta_k} - \frac{\partial^2 f(\boldsymbol{\beta})}{\partial \theta_j \partial\theta_k} \right| \nonumber \\ &\leq M L_2' \|\alpha - \beta\|_2 = 8 M^{3/2} H_{\max}^3 \|\alpha - \beta\|_2
\end{align}
and therefore $f$ is $L_2$-second-order smooth with $L_2 = 8 M^{3/2} H_{\max}^3$.
\end{proof}

% \tyl{Fix the overfill in Eqs. (16), (20), (28), (29).}

Based on above lemma, we obtain the following theorem:
\begin{theorem}[Smoothness]\label{thm:theo5}
The loss function $\mathcal{L}(\boldsymbol\theta)$ in the training of QAEGate model is $L_1$-smooth and $L_2$-second-order smooth.
\end{theorem}

\begin{proof}
Let $p(\boldsymbol\theta)$ be the probability of outcome $|+\rangle$ in the SWAP test when the gates are parameterized with $\boldsymbol{\theta=(\theta_{\rm le},\theta_{\rm re},\theta_{\rm ld},\theta_{\rm rd})}$.
Because of $\mathcal{L}(\boldsymbol \theta) = 1-f(\boldsymbol \theta)$ and $f(\boldsymbol \theta) = 2p(\boldsymbol\theta)-1$, we just need to discuss the smoothness of $p(\boldsymbol\theta)$. Since $p(\boldsymbol\theta)$ is the probability of a measurement outcome of a unitary circuit, and in this circuit, every variable gate has the form of $U_\theta = \exp(i\theta H)$, we can apply \lem{lip_smooth}. Thus, $p(\boldsymbol\theta)$ is $L_1$-smooth and $L_2$-second-order smooth where $L_1 = 4MH_{\max}^2, L_2 = 8M^{3/2}H_{\max}^3$ and $M$ is the number of variable gates in the circuit. To prove \thm{theo5}, we only need to show that $H_{\max}$ is bounded.

The variable gates in our circuit only includes the following single- and two-qubit gates:
\begin{align}
	&R_X(\theta) := \exp(i\theta \sigma_x),~R_Y(\theta) := \exp(i\theta \sigma_y),~R_Z(\theta) := \exp(i\theta \sigma_z) \\
	&XX(\theta):=\exp(i\theta \sigma_x\otimes \sigma_x),~YY(\theta):=\exp(i\theta \sigma_y\otimes \sigma_y),~ZZ(\theta):=\exp(i\theta \sigma_z\otimes \sigma_z)
\end{align}
Therefore, $H_{\max} = \max\{\|\sigma_x\|_2,\|\sigma_y\|_2,\|\sigma_z\|_2,\|\sigma_x\otimes \sigma_x\|_2,\|\sigma_y\otimes \sigma_y\|_2,\|\sigma_z\otimes \sigma_z\|_2\} = 1$.
\end{proof}
\begin{proof}
Based on \thm{theo5}, we have proved that $\mathcal{L}(\boldsymbol\theta)$ is $L_1$-smooth, where $L_1 = 4M$. In our implementation, there is only one variable for each variable unitary gate, so $M=\dim{\boldsymbol \theta}$ here. Then we can apply \lem{converge} and obtain that the convergence rate is $T=\mathcal{O}(\frac{4\dim{\boldsymbol \theta}}{\epsilon^4})$ for finding $\mathbb{E}[\Vert \nabla_{\boldsymbol{\theta}} \mathcal{L}(\boldsymbol{\theta}) \Vert^2]\leq \epsilon^2$.

\end{proof}

% \section{First section of the appendix}
% Quantum allows the usage of appendices.

% \subsection{Subsection}
% Ideally, the command \texttt{\textbackslash{}appendix} should be put before the appendices to get appropriate section numbering.
% The appendices are then numbered alphabetically, with numeric (sub)subsection numbering.
% Equations continue to be numbered sequentially.
% \begin{equation}
%   A \neq B
% \end{equation}
% You are free to change this in case it is more appropriate for your article, but a consistent and unambiguous numbering of sections and equations must be ensured.

% If you want your appendices to appear in \texttt{onecolumn} mode but the rest of the document in \texttt{twocolumn} mode, you can insert the command \texttt{\textbackslash{}onecolumn\textbackslash{}newpage} before \texttt{\textbackslash{}appendix}.

% \section{Problems and Bugs}
% In case you encounter problems using the quantumarticle class please analyze the error message carefully and look for help online; \href{http://tex.stackexchange.com/}{http://tex.stackexchange.com/} is an excellent resource.
% If you cannot resolve a problem, please open a bug report in our bug-tracker under \href{https://github.com/quantum-journal/quantum-journal/issues}{https://github.com/quantum-journal/quantum-journal/issues}.
% You can also contact us via email under \href{mailto:latex@quantum-journal.org}{latex@quantum-journal.org}, but it may take significantly longer to get a response.
% In any case, we need the full source of a document that produces the problem and the log file showing the error to help you.

\end{document}